\documentclass[envcountsame,envcountsect]{llncs}
\usepackage[english]{babel}
\usepackage[utf8x]{inputenc}
\usepackage[T1]{fontenc}
\usepackage[a4paper,top=3cm,bottom=2.5cm,left=3cm,right=3cm,marginparwidth=1.75cm]{geometry}
\usepackage{fullpage}

\usepackage{tikz}
\usepackage{enumitem}
\usepackage[skins]{tcolorbox}
\usepackage{amsmath}
\usepackage{mathtools}
\usepackage{graphicx}
\usepackage{amssymb}
\usepackage{amsfonts}
\usepackage[colorinlistoftodos]{todonotes}
\usepackage{breakcites}
\usepackage{upgreek}
\usepackage{bbold}
\usepackage{dsfont}
\usepackage{stmaryrd}
\usepackage{wasysym}
\usepackage{framed}
\usepackage{empheq}
\usepackage{xcolor}
\usepackage{theorem}
\usepackage{algpseudocode}
\usepackage{refcount}
\usepackage{qcircuit}
\usepackage{makecell}

\usepackage{float}
\usepackage{xspace}

\usepackage[colorlinks=true, allcolors=blue,breaklinks=true]{hyperref}

\usepackage[normalem]{ulem}

\def\fullVersion{1}      									

\newcommand{\onlyshortversion}[1]{\if\fullVersion1\else{#1}\fi}
\newcommand{\onlyfullversion}[1]{\if\fullVersion1#1\fi}
\newcommand{\shortorfullversion}[2]{\if\fullVersion1{#2}\else{#1}\fi}

\newcommand{\bra}[1]{\langle #1|}
\newcommand{\ket}[1]{|#1\rangle}

\newcommand{\braket}[2]{\langle #1|#2\rangle}

\definecolor{dgreen}{rgb}{.1,.5,.1}
\definecolor{grey}{rgb}{.6,.6,.6}

\newcommand{\proj}[1]{\ket{#1}\!\bra{#1}}

\usepackage{textgreek}
\newcommand{\sigp}{\textSigma-protocol\xspace}
\newcommand{\sigps}{\textSigma-protocols\xspace}

\def\bfbot{{\boldsymbol\bot}}

\newcommand{\instance}{\ensuremath{\textit{\textsf{inst}}}}
\def\INST{{\cal I}}

\newcommand{\I}{\mathbb I}

\newcommand{\C}{\mathbb C}

\newcommand{\N}{\mathbb N}

\newcommand{\bfx}{\mathbf{x}}
\newcommand{\bfy}{\mathbf{y}}

\newcommand{\DB}{\mathfrak{D}}

\newcommand{\cO}{{\sf cO}}
\newcommand{\CL}{{\sf CL}}
\newcommand{\SZ}[1][s]{{\sf SZ}_{\leq #1}}
\newcommand{\SUC}{{\sf SUC}\xspace}

\def\spc{\hspace{0.05ex}}
\def\nspc{\hspace{-0.1ex}}
\newcommand{\QTC}[3][]{\big\llbracket\spc#2\nspc\stackrel{#1}{\rightarrow}\nspc#3\spc\big\rrbracket}

\newcommand{\qQTC}[3][q]{\big\llbracket\spc#2\stackrel{#1\spc}{\Longrightarrow}\nspc#3\spc\big\rrbracket}

\renewcommand{\P}{\mathsf{P}}
\renewcommand{\L}{\mathsf{L}}

\newcommand{\trafo}{transformation\xspace}
\newcommand{\Trafo}{Transformation\xspace}
\newcommand{\FST}{\shortorfullversion{FS \trafo}{Fiat-Shamir \trafo}\xspace}
\newcommand{\FS}{\shortorfullversion{FS}{Fiat-Shamir}\xspace}
\newcommand{\PoKOE}{\shortorfullversion{PoK-OE}{proof of knowledge with online extractability}\xspace}

\newcommand{\MTree}{\mathsf{MTree}}
\newcommand{\MRoot}{\mathsf{MRoot}}
\newcommand{\MAuth}{\mathsf{MAuth}}
\newcommand{\MOcto}{\mathsf{MOcto}}

\newcommand{\Tree}{\mathsf{Tree}}

\newcommand{\Auth}{\mathsf{Auth}}
\newcommand{\Octo}{\mathsf{Octo}}

\newcommand{\OctoVerify}{\mathit{Octo\!Verify}}

\def\lf{\mathrm{lf}}

\def\showcomments{1}          

\def\serge#1{\ifnum\showcomments=1{\color{red}\sf [SF: #1]}\fi}
\def\jelle#1{\ifnum\showcomments=1{\color{orange}\sf [JD: #1]}\fi}
\def\cs#1{\ifnum\showcomments=1{\color{violet}\sf [CS: #1]}\fi}
\def\cm#1{\ifnum\showcomments=1{\color{dgreen}\sf [CM: #1]}\fi}

\newcommand{\CnO}{C\&O\xspace}
\newcommand{\MCnO}{Merkle-tree-based C\&O\xspace}

\definecolor{aqua}{rgb}{0.0, 1.0, 1.0}

\DeclareMathSymbol{\shortminus}{\mathbin}{AMSa}{"39}

	\pgfdeclarelayer{nodelayer}
	\pgfdeclarelayer{edgelayer}
	\pgfsetlayers{edgelayer,nodelayer}
		\tikzstyle{arrow}=[->]
	\tikzstyle{arrow left}=[<-]
	\tikzstyle{Tilted label}=[rotate=-90]

\title{Efficient NIZKs and Signatures from Commit-and-Open Protocols in the QROM}

\author{Jelle Don\inst{1} \and  Serge Fehr\inst{1,2} \and Christian Majenz\inst{3}\and Christian Schaffner \inst{4,5}}

\institute{
	Centrum Wiskunde \& Informatica (CWI), Amsterdam, Netherlands \and 
	Mathematical Institute, Leiden University, Netherlands \and
		Cyber Security Section, Department of Applied Mathematics and Computer Science, Technical University of Denmark, Kgs. Lyngby, Denmark \and
	Informatics Institute, University of Amsterdam, Amsterdam, Netherlands\and
	QuSoft, Amsterdam, Netherlands 
	\\ \email{jelle.don@cwi.nl}, \email{serge.fehr@cwi.nl}, \email{chmaj@dtu.dk}, \email{c.schaffner@uva.nl}
}

\begin{document}
	
	\maketitle

	\setcounter{footnote}{0}

	\begin{abstract}
		Commit-and-open \sigps are a popular class of protocols for constructing non-interactive zero-knowledge arguments and digital-signature schemes via the Fiat-Shamir \trafo. Instantiated with hash-based commitments, the resulting non-interactive schemes enjoy tight online-extractability in the random oracle model. Online extractability improves the tightness of security proofs for the resulting digital-signature schemes by avoiding lossy rewinding or forking-lemma based extraction.  
		
		In this work, we prove tight online extractability in the quantum random oracle model (QROM), showing that the construction supports post-quantum security. First, we consider the default case where committing is done by element-wise hashing. In a second part, we extend our result to Merkle-tree based commitments. Our results yield a significant improvement of the provable post-quantum security of the digital-signature scheme Picnic.
                
		Our analysis makes use of a recent framework by Chung et al.~\cite{CFHL21} for analysing quantum algorithms in the QROM using purely classical reasoning. Therefore, our results can to a large extent be understood and verified without prior knowledge of quantum information science.
	\end{abstract}
	
	%
	%
	%

	\section{Introduction}

	


	Some interactive proofs come with amazing properties like \emph{zero-knowledge} which intuitively allows a prover to convince a verifier that she knows the witness to an NP-statement without giving away \onlyfullversion{any }information about this witness. Such zero-knowledge proofs of knowledge are some of the most fascinating objects in cryptography, and possibly in all of theoretical computer science. One might suspect that their ``magic'' \shortorfullversion{is due to the prover and verifier running }{is rooted in the fact that the prover and verifier run}  an \emph{interactive} protocol with each other, and that this interaction causes the verifier to be convinced. Surprisingly, if the interactive proof is of suitable form, e.g. a \sigp (i.e. a $3$-round public-coin protocol), the Fiat-Shamir \trafo~\cite{Fiat1987} provides a natural way to remove the interaction from such protocols while preserving (most of) the security properties, resulting in \emph{non-interactive zero-knowledge} proofs (NIZKs). 
The idea is to compute the challenge $c$ as a hash $c = H(a)$ of the first message, rather than letting the verifier choose $c$. 
If the original \sigp has additional soundness properties, the resulting NIZK after the Fiat-Shamir \trafo is ideally suited to \shortorfullversion{construct}{be turned into} a \emph{digital-signature scheme}, simply by hashing the message $m$ to be signed together with the first message $a$ in order to obtain the challenge $c$. The candidates Picnic \cite{Chase2017} and Dilithium \cite{NIST:DILITHIUM} in the ongoing NIST post-quantum cryptography competition follow this design paradigm.

This intuitive preservation of security properties under the Fiat-Shamir \trafo can be formalized in the random-oracle model (ROM), where the hash function $H$ is treated as a uniformly random function, and the security reduction gets \emph{enhanced access} to anybody who queries the random oracle, by seeing which values are queried, and by possibly returning (random-looking) outputs. While this situation is conveniently easy to handle in a non-quantum world, complications arise in the context of post-quantum security. When studying the security of these non-quantum protocols against attackers equipped with large-enough quantum computers, it is natural to assume that such attackers have access to the public description of the employed hash function, and can therefore compute
it in superposition on their quantum computers. Therefore, the proper notion of post-quantum security for random oracles is the \emph{quantum-accessible random-oracle model (QROM)} as introduced in~\cite{Boneh2011}. Due to the difficulty of recording adversarial random-oracle queries in superposition (also referred to as the \emph{recording barrier}), establishing post-quantum security in the QROM has turned out to be quite a bit more difficult compared to the regular ROM.

Previous results 
 in \cite{DFMS19} (and concurrently in~\cite{Liu2019}) establish that for any interactive \sigp $\Pi$ that is a proof of knowledge, the non-interactive $\mathsf{FS}[\Pi]$ is a proof of knowledge in the QROM. \cite{DFM20} simplified the technical proof and extended these results to multi-round interactive proofs. However, the most desirable property from such a proof  of knowledge is \emph{online extractability}. 
 	Indeed, online extractability avoids {\em rewinding}, which typically causes a significant loss in the security reduction (see later for a comparison) and has other disadvantages. 
Thus, 
 online extractability allows for the tightest security reductions.

Chailloux was the first to aim for showing online extractability of the Fiat-Shamir transformation in the QROM 
when considering the relevant 
 class of {\em commit-and-open} (\CnO) \sigps and modelling the hash function used for the commitments (and for computing the challenge) as a random oracle. Indeed, the Fiat-Shamir \trafo of such \CnO \sigps are known to be online extractable in the classical ROM (see e.g. discussion in \cite{Fischlin05})
 . 
 In a first attempt~\cite{Chailloux19}, Chailloux tried to lift the argument to the quantum setting by means of Zhandry's compressed-oracle technique~\cite{Zhandry2018}, which offers a powerful approach for re-establishing ROM results in the QROM, that has been successful in many instances. Unfortunately, this first attempt contained a subtle flaw, which turned out to be unfixable, and despite changing the technical approach, the latest version~\cite{Chailloux20} of this work still contains an open gap in the proof, which is put as an assumption.%
\footnote{Informally, quoting from~\cite{Chailloux20}, the considered Assumption~2 is that the random oracle can be replaced with a random function of a particular form {\em ``without harming too much the studied scheme''}. More formally, the security loss caused by the considered replacement is assumed to remain bounded by a given function of the number of oracle queries. This assumption is rather ad-hoc and non-standard in that it is very much tailored to the scheme and its proof.  
Furthermore, even though Assumption~2 
is an assumption that could potentially be proven in future work
, it is hard to judge whether proving the assumption is actually any easier than proving the security of the considered scheme {\em directly}, avoiding Assumption~2\,---\,as a matter of fact, in this work we show that the latter is feasible, while Assumption~2 remains open.} 

In a recent article~\cite{DFMS21}, 
 online extractability of {\em interactive} 
 \CnO \sigps $\Pi$ in the QROM is established
 	; the result applies as soon as $\Pi$ satisfies some liberal notion of {\em special soundness}, which is typically satisfied
 . As pointed out in Appendix~E of~\cite{DFMS21}, one can use previous results from~\cite{DFMS19,Liu2019,DFM20} to reduce the extractability of the resulting non-interactive protocol ${\sf FS}[\Pi]$ to the extractability of the interactive protocol $\Pi$. However, the resulting extraction error still scales as $O(\varepsilon/q^2)$, which results in a prohibitive loss for digital-signature schemes (see Table~\ref{tab:comparison}), leaving open the main question originally posed by Chailloux:
\begin{quote}
\emph{How to establish tight security reductions of the Fiat-Shamir \trafo for commit-and-open \sigps in the QROM?}
\end{quote}
As the technical quantum details of Zhandry's 
compressed-oracle technique are rather complicated and only accessible for experts, a recent article by Chung, Fehr, Huang and Liao~\cite{CFHL21} \onlyfullversion{attempts to give a comprehensive exposition of Zhandry's technique. In addition, they }establish a framework that allows researchers without extensive quantum knowledge to still deploy the compressed-oracle technique (in certain cases), basically by reasoning about classical quantities only. In short, the punchline of~\cite{CFHL21} is that, if applicable, one can prove {\em quantum} query complexity lower bounds (think of collision finding, for instance) by means of the following recipe, which is an abstraction of the technique developed in a line of works started by Zhandry \cite{Zhandry2018,LZ19,CGLQ20,HM20}. First, one considers the corresponding {\em classical} query complexity problem, analyzing it by simulating the random oracle using lazy sampling and showing that the database, which keeps track of the oracle queries and the responses, is unlikely to satisfy a certain property (e.g. to contain a collision) after a bounded number of queries. Then, one lifts the analysis to the quantum setting by plugging \onlyfullversion{certain }key observations from the classical analysis into generic theorems provided by the \cite{CFHL21} framework.

	\subsection{Our Contributions}

	In this work, we \onlyfullversion{slightly} extend the framework from \cite{CFHL21}, and use it in a conceptually new (and arguably roundabout) way to establish strong and tight security statements for a large, popular class of non-interactive zero-knowledge proofs and digital signature schemes. 
        In broad strokes, our contributions are threefold.
	
	\subsubsection{Online extractability for a class of NIZKs in the QROM.} 
		We prove online extractability of the Fiat-Shamir transformation in the QROM for (a large class of) \CnO \sigps
		. This solves the problem considered and attacked by Chailloux. In more detail, 
	we prove that if 
		the considered \CnO \sigp satisfies some very liberal notion of special soundness
	, then the resulting NIZK is a proof of knowledge with online extractability in the QROM, 
		i.e., when the hash function used for the commitments and the \FST is modeled as a quantum-accessible random oracle. 
	Our security reduction is tight: Whenever a prover outputs a valid proof, the online-extractor succeeds, except with a small probability accounting for collision and preimage attacks on the involved hash functions. For previous reductions, the guaranteed extraction success probability was at least by a factor of $q^2$ smaller than the succes probability of the prover subjected to extraction (see Table~\ref{tab:comparison}). This is our main technical contribution, see Theorem \ref{thm:extractor}.  Our result also applies to a variant of the Fiat-Shamir \trafo where a digital signature scheme (DSS) is constructed. It thereby, for the first time, enables a multiplicatively tight security reduction for, e.g., DSS based on the MPC-in-the-head paradigm \cite{IKOS07}, like Picnic \cite{Chase2017}, Banquet \cite{CdSGKOSZ21} and Rainier \cite{DKRSZ21}, in the QROM.
	

	\subsubsection{A more efficient Unruh \trafo.}  When a 
		\sigp 
	does not have the mentioned \CnO structure, a non-interactive proof of knowledge with online extractability in the QROM can be obtained using the Unruh \trafo \cite{Unruh2015}. 
		For technical reasons, the Unruh \trafo requires the hash function to be {\em length preserving}, which may result in large commitments, and thus large NIZKs and digital signature schemes. We revisit this \trafo and show, by a rather direct application of our main result above, that the online extractability of the Unruh transform still holds when using a {\em compressing} hash function. The crucial observation is that the Unruh \trafo can be viewed as the composition of a pre-Unruh \trafo, which makes use of hash-based commitments and results in a \CnO protocol, and the Fiat-Shamir \trafo. By applying our security reduction, we obtain the tight online extractability without requiring the hash function to be length preserving. 

	\subsubsection{More efficient NIZKs via Merkle tree based commitments.}  In real-world constructions based on \CnO protocols, like e.g., the Picnic digital signature scheme, commitments and their openings are responsible for a significant fraction of the signature/proof size. For certain parameters, this cost can be reduced by using a collective commitment mechanism based on Merkle trees. This was observed in passing, e.g. in \cite{Fischlin05}, and is exploited in the most recent versions of Picnic. We formalize \MCnO protocols and extend our main result to NIZKs constructed from them (see Theorem \ref{thm:extractor-M}). Applications \onlyfullversion{of this result} include a security reduction of Picnic 3, the newest version of the Picnic digital signature scheme, that is significantly tighter than existing ones: An adversary against the Picnic~3 signature scheme in the QROM with success probability $\varepsilon$ can now be used to break the underlying hard problem with probability $\varepsilon$, up to some additive error terms, while previous reductions yielded at most $\varepsilon^5/q^{10}$, where $q$ is the number of random oracle queries. We outline this reduction in Section \ref{sec:picnic}.
	
	We compare our reductions in detail to existing techniques in Table \ref{tab:comparison}.
	
	\begin{table}
		\begin{center}
			\begin{tabular}{c||c|c|c}
				& 2-s$\Rightarrow$PoK&\makecell{ PoK$\stackrel{\mathrm{FS}}{\Rightarrow}$NIZK-PoK,\\ PoK$\stackrel{\mathrm{FS}}{\Rightarrow}$UF-NMA DSS}&\makecell{2-s$\stackrel{\mathrm{FS}}{\Rightarrow}$NIZK-PoK,\\2-s$\stackrel{\mathrm{FS}}{\Rightarrow}$UF-NMA DSS}\\
				\hline 
				\hline
				\makecell{Unruh rewinding \cite{Unruh2012}\\ + generic FS \cite{DFMS19}}& $O(\varepsilon^3)$& $O(\varepsilon/q^2)$&$O(\varepsilon ^3/q^6)$\\
				\hline
				\makecell{\textcolor{grey}{\sigp OE \cite{DFMS21}} \\+ generic FS \cite{DFMS19}}& \textcolor{grey}{$\varepsilon-g(q,r,n)$}& $O(\varepsilon/q^2)$& \textcolor{grey}{$O(\varepsilon/q^2)-g(q,r,n)$}\\
				\hline
				\makecell{\textbf{this work}:\\ \textbf{NIZK OE}}&-&-& $
					\boldsymbol{\varepsilon-h(q,r,n)}
				$
			\end{tabular}
		\end{center}
	\caption{Comparison of the losses of different reductions for the construction of a NIZK proof of knowledge (NIZK-PoK) from a special-sound (Merkle tree based) \CnO protocol with constant challenge space size $C$ using $r$-fold parallel repetition and the Fiat-Shamir \trafo. ``OE'' stands for online extraction, 2-s for special soundness, UF-NMA for plain unforgeability and DSS for digital signature scheme. If the content of a cell in row ``security property A $\Rightarrow$ security property B'' is $f(\varepsilon)$, this means that an adversary breaking property B with probability $\varepsilon$ yields an adversary breaking property A with probabilty $f(\varepsilon)$. \textcolor{grey}{Grey text} indicates results that do  not apply to \MCnO protocols like the one used to construct the digital signature schemes Picnic 2 \cite{KZ20} and Picnic 3 \cite{Picnic}.  The additive error terms are $g(q,r,n)=C^{-r}+O(rq2^{-n/2})+O(q^32^{-n})$ and $h(q,r,n)=O(q^32^{-n})+O(q^2C^{-r})$, where $n$ is the output length of the random oracles, and $q$ is the number of adversarial (quantum) queries to the random oracle.  Finally, we note that the constants hidden by the big-O in $h(q,r,n)$ are reasonable, see Theorems~\ref{thm:extractor} and~\ref{thm:extractor-M}. \vspace{-3ex}
}\label{tab:comparison}
	\end{table}
	
	\subsection{Technical Overview}


Our starting point is the fact that the compressed-oracle technique can be \shortorfullversion{seen}{appreciated} as a variant of the classical lazy-sampling technique that is applicable in the QROM. 
Namely, to some extent and informally described here, the compressed-oracle technique gives access to a database that contains the hash values that the adversary $\cal A$, who has interacted with the random oracle (RO), may know. In particular, up to a small error, for any claimed-to-be hash value $y$ output by $\cal A$, one can find its preimage $x$ by inspecting the database (and one can safely conclude that $\cal A$ does not know a preimage of $y$ if there is none in the database). 
Recalling that a \CnO \sigp $\Pi$ is an interactive proof where the first message consists of hash-based commitments, and exploiting that typically some sort of special soundness property ensures that knowing sufficiently many preimages of these commitments/hashes allows one to efficiently compute a witness, constructing an online extractor for the Fiat-Shamir \trafo ${\sf FS}[\Pi]$ then appears straightforward:
%
%
The extractor $\cal E$ simply runs the (possibly dishonest) prover $P^*$, answering RO queries using the compressed oracle. Once $P^*$ has finished and outputs a proof, $\cal E$ measures the compressed-oracle database and classically reads off any preimages of the commitments in the proof. 
Finally, $\cal E$ run the special soundness extractor that computes a witness from 
	the obtained preimages
. It is, however, not obvious that the database contains the preimages of the commitments that are {\em not} opened in the proof, or that these preimages are correctly formed. Intuitively this should be the case: the RO used for the Fiat-Shamir \trafo replaces interaction in that it forces the prover to chose a full set of commitments \emph{before} knowing which ones need to be opened. 
	The crux lies in replacing this intuition by a rigorous proof
. 

The main insight leading to our proof 
 is that the event that needs to be controlled, namely that \emph{the prover succeeds yet the extractor fails}, can be translated into a property $\SUC$ (as in ``adversarial SUCcess'' ) of the compressed-oracle database, which needs to be satisfied for the event to hold. It is somewhat of a peculiar property though. The database properties that have led to query complexity lower bounds in prior work, e.g. for (multi-)collision finding \cite{LZ19,HM20,CFHL21} and similar problems \cite{Zhandry2018,CGLQ20,BLZ21},
require the database to contain some particular input-output pairs (e.g. pairs that collide), while the database property $\SUC$ additionally {\em forbids} certain input-output pairs to be contained.

Indeed, the framework from \cite{CFHL21} is almost expressive enough to treat our problem. So, after a mild extension, we can apply it to prove that it is hard for any query algorithm to cause the compressed-oracle database to have property $\SUC$. Analyzing the relevant classical statistical properties of $\SUC$ is somewhat tedious but can be done (see the proof of Lemma \ref{lem:SUC-Oc-CLbound}). 
	The resulting bound on the probability for the database to satisfy $\SUC$ then gives us a bound on the probability of the event that the prover succeeds in producing a valid proof while at the same time fooling the extractor.


Whenever it is advantageous for communication complexity, a Merkle tree can be used to collectively commit to all required messages in a \CnO protocol. 
 This collective commitment is one of the optimizations that improve the performance of, e.g. Picnic 2 \cite{KZ20} over Picnic \cite{Chase2017}. As the above-described argument for the extractability of \CnO protocols already analyses iterated hashing (the hash-based commitments are hashed to compute the challenge), it generalizes to \MCnO protocols without too much effort. We present this generalization in Section \ref{sec:MCnO}, and obtain similar bounds (see Theorem \ref{thm:extractor-M}).

\subsection{Additional Related Work} \label{sec:relatedwork}
Besides the already mentioned work above, we note that Chiesa, Manohar and Spooner~\cite{CMS19} consider and prove security of various SNARG constructions, while we consider the Fiat-Shamir \trafo of \CnO protocols with a form of special soundness. 
	Similar in to \cite{CFHL21}, they also provide some tools for deducing security of certain oracle games against quantum attacks by bounding a natural classical variant of the game.

	\section{Preliminaries}\label{sec:Prelim}
	
	Our main technical proofs reliy on the recently introduced framework by Chung, Fehr, Huang, and Liao~\cite{CFHL21} for proving query complexity bounds in the QROM. This framework exploits Zhandry's compressed-oracle technique but abstracts away all the quantum aspects, so that the reasoning becomes purely classical. We give here an introduction to a simplified, and slightly adjusted version that does not consider parallel queries. We start with recalling (a particular view on) the compressed oracle. 
	Along the way, we also give an improved version of Zhandry's central lemma for the compressed oracle. 
	
	Before getting into this, we fix the following standard notation. For any positive integer $\ell>0$, we set $[\ell]:=\{1,2,\ldots, \ell\}$, and we let $2^{[\ell]}$ denote the power set of $[\ell]$, i.e., the set of all subsets of $[\ell]$. 
	\onlyfullversion{We write $\{0,1\}^{\leq \ell}$ for the set of bit strings of size at most $\ell$, including the empty string denoted $\emptyset$; similarly for $\{0,1\}^{< \ell}$. Concatenation of two bit strings $v \in \{0,1\}^m$ and $w \in \{0,1\}^n$ is denoted by $v\|w \in \{0,1\}^{m+n}$. } 
	
	Finally, for any finite non-empty set $\cal Z$, $\C[{\cal Z}]$ denotes the Hilbert space $\C^{|{\cal Z}|}$ together with a basis $\{\ket{z}\}$ labeled by the elements $z \in \cal Z$.

	\subsection{The Compressed Oracle\,---\,Seen as Quantum Lazy Sampling}
	
	With the goal to analyze oracle algorithms that interact with a RO $H: {\cal X} \to {\cal Y}$, consider the set $\DB$ of all functions $D: {\cal X} \to {\cal Y} \cup \{\bot\}$, where $\bot$ is a special symbol. Such a function is referred to as a {\em database}. Later, we will fix ${\cal X} = \{0,1\}^{\leq B}$ and ${\cal Y} = \{0,1\}^n$. For $D \in \DB$, $x \in {\cal X}$ and $y \in {\cal Y} \cup \{\bot\}$, $D[x \!\mapsto\! y]$ denotes the database that maps $x$ to $y$ and otherwise coincides with $D$, i.e., $D[x \!\mapsto\! y](x) = y$ and $D[x \!\mapsto\! y](\bar x) = D(\bar x)$ for all $\bar x \in {\cal X} \setminus \{x\}$. 
	
	Following the exposition of~\cite{CFHL21}, the compressed-oracle technique is a quantum analogue of the classical lazy-sampling technique, commonly used to analyze algorithms in the classical ROM. In the classical lazy-sampling technique, the (simulated) RO starts off with the empty database, i.e., with $D_0 = \bfbot$, which maps any $x \in \cal X$ to $\bot$. Then, recursively, upon a query~$x$, the current database $D_i$ is updated to $D_{i+1} := D_i$ if $D_i(x) \neq \bot$, and to $D_{i+1} := D_i[x \!\mapsto\! y]$ for a randomly chosen $y \in \cal Y$ otherwise. This construction ensures that $|\{ x \,|\, D_i(x) \!\neq\! \bot \}| \leq i$; 
	after $i$ queries 
	thus, using standard sparse-encoding techniques, the database $D_i$ can be efficiently represented and updated. 
	
	In the compressed-oracle quantum analogue of this lazy-sampling technique, the (simulated) RO also starts off with the empty database, but now considered as a quantum state $\ket{\bfbot}$ in the $|\DB|$-dimensional state space $\C[\DB]$, and after $i$ queries the state of the compressed oracle is then supported by databases $\ket{D_i}$ for which $|\{ x \,|\, D_i(x) \!=\! \bot \}| \leq i$.%
	\footnote{This means that the density operator that describes the state of the compressed oracle has its support contained in the span of these $\ket{D_i}$. }
	Here, the update is given by a unitary operator $\cO$ acting on $\C[{\cal X}] \otimes \C[{\cal Y}]\otimes \C[\DB]$, i.e., on the query register, the response register, and the state of the compressed oracle. With respect to the computational basis $\{\ket{x}\}$ of $\C[{\cal X}]$  and the Fourier basis $\{\ket{\hat y}\}$ of $\C[{\cal Y}]$, $\cO$ is a {\em control} unitary, i.e., of the form $\cO = \sum_{x,\hat y} \proj{x} \otimes \proj{\hat y} \otimes \cO_{x,\hat y}$, where $\cO_{x,\hat y}$ is a unitary on $\C[{\cal Y} \cup \{\bot\}]$, which in the above expression is understood to act on the register that carries the value of the database at the point $x$. More formally, $\cO_{x,\hat y}$ acts on register $R_x$ when identifying $\C[\DB]$ with $\bigotimes_{x \in \cal X} \C[{\cal Y} \cup \{\bot\}]$ by means of the isomorphism $\ket{D} \mapsto \bigotimes_{x \in \cal X}\ket{D(x)}_{R_x}$. We refer to Lemma~4.3 in the full version of \cite{CFHL21} for the full specification of $\cO_{x,\hat y}$; it is not really relevant here. 
	
	The compressed oracle is tightly related to the {\em purified} oracle, which initiates its internal state with a uniform superposition $\sum_h \ket{H} \in \C[\DB]$ of all functions $H: {\cal X} \to {\cal Y}$, and then answers queries ``in superposition''. Indeed, at any point in time during the interaction with an oracle quantum algorithm $\cal A$, the joint state of $\cal A$ and the compressed oracle coincides with the joint state of $\cal A$ and the purified oracle after ``compressing'' the latter.%
	\footnote{The terminology is somewhat misleading here; the actual compression takes place when invoking the sparse encoding (see below). }
	Formally, identifying $\C[\DB]$ with $\bigotimes_{x \in \cal X} \C[{\cal Y} \cup \{\bot\}]$ again, the compression of the state of the purified oracle works by applying the unitary ${\sf Comp}$ to each register $R_x$, where 
	$$
	{\sf Comp}: \ket{y} \mapsto (\ket{y} +\shortorfullversion{(\ket{\bot} - \ket{\hat 0})/\sqrt{|\mathcal Y|}}{\frac{1}{\sqrt{|\mathcal Y|}}(\ket{\bot} - \ket{\hat 0})}
	$$
	for any $y \in \cal Y$, and ${\sf Comp}: \ket{\bot} \mapsto \ket{\hat 0}$. Here, $\ket{\hat 0}$ is the $\hat 0$-vector from the Fourier basis $\{\ket{\hat y}\}$ of $\C[{\cal Y}]$. 
	
	Similarly to the classical case, by exploiting a quantum version of the sparse-encoding technique, both the internal state of the compressed oracle and the evolution $\cO$ can be efficiently computed. Furthermore, for any classical function $f: \DB \to {\cal T}$ that can be efficiently computed when given the sparse representation of $D \in \DB$, the corresponding quantum measurement given by the projections $P_t = \sum_{D: f(D)=t} \proj{D}$ can be efficiently performed when given the sparse representation of the internal state of the compressed oracle. In particular, in Lemma~\ref{lemma:zha} below, the condition ${\bf y} = D({\bf x})$ for given $\bf x$ and $\bf y$ can be efficiently checked by a measurement. See Appendix~A in (the full version of) \cite{CFHL21}, or Appendix~B in \cite{DFMS21}
	for more details on this technique. 
	
	In the classical lazy-sampling technique, if at the end of the execution of an oracle algorithm $\cal A$, having made $q$ queries to the (lazy-sampled) RO, the database $D_q$ is such that, say, $D_q(x) \neq 0$ for any $x \in \cal X$, then $\cal A$'s output is unlikely to be a $0$-preimage, i.e., an $x$ that is hashed to $0$ upon one more query. $\cal A$'s best chance is to output an $x$ that he has not queried yet, and thus $D_q(x) = \bot$, and then he has a $1/|{\cal Y}|$-chance that $D_{q+1}(x) := D_{q}[x \!\mapsto\! y](x) = 0$, given that $y$ is randomly chosen. Something similar holds in the quantum setting, with some adjustments. The general statement is given by the following result by Zhandry.%
	
	\begin{lemma}[Lemma~5 in \cite{Zhandry2018}]\label{lemma:zha}
		Let $R \subseteq {\cal X}^\ell \times {\cal Y}^\ell \times {\cal Z}$ be a relation, and let $\cal A$ be an oracle quantum algorithm that outputs ${\bf x} \in {\cal X}^\ell$, ${\bf y} \in {\cal Y}^\ell$ and $z \in \cal Z$. 
		Furthermore, let 
		$$
		p = p({\cal A}) :=  \Pr[ {\bf y} \!=\! H({\bf x}) \wedge ({\bf x},{\bf y},z) \!\in\! R ] 
		$$ 
		be the considered probability when $\cal A$ has interacted with the standard RO, initialized with a uniformly random function $H$, and let 
		\begin{align*}
			p' = p'({\cal A}) &:=  \Pr[ {\bf y} \!=\! D({\bf x}) \wedge ({\bf x},{\bf y},z) \!\in\! R ]  
		\end{align*} 
		be the considered probability when $\cal A$ has interacted with the compressed oracle instead and $D$ is obtained by measuring its internal state (in the basis $\{\ket{D}\}_{D \in \DB}$). Then
		$$
		\sqrt{p} \leq \sqrt{p'} + \shortorfullversion{\sqrt{\ell/|{\cal Y}|}}{\sqrt\frac{\ell}{|{\cal Y}|} }\, . 
		$$
	\end{lemma}
	
	\onlyfullversion{
	\begin{remark}\label{rem:Z=0}
		This bound is particular useful in case ${\cal Z} = \emptyset$ (or $R$ does not depend on its third input~$z$), since then $p'$ is bounded by $\Pr[ \exists \, \tilde{\bf x}: (\tilde{\bf x},D(\tilde{\bf x})) \!\in\! R ]$ and the latter is determined solely by the evolution of the compressed oracle (when interacting with $\cal A$) and does not depend on the actual output of $\cal A$. 		
	\end{remark}
	}
	
	In \shortorfullversion{Appendix~\ref{app:Corollaries}}{Section~\ref{sec:ImprovedLemma}}, Corollary~\ref{cor:link}, we will give an alternative such relation between the success probability of an algorithm interacting with the actual RO, and probabilities obtained by inspecting the compressed oracle instead. Strictly speaking, the results of Lemma~\ref{lemma:zha} and Corollary~\ref{cor:link} are incomparable, but in typical applications the latter gives a significantly better bound.

	\subsection{The Quantum Transition Capacity and Its Relevance}

	The above discussion shows that, in order to bound the success probability $p$ of an oracle algorithm $\cal A$, it is sufficient to bound the probability of the database $D$, obtained by measuring the internal state of the compressed oracle after the interaction with $\cal A$, satisfying a certain property (e.g., the property of there existing an $x$ such that $D(x) = 0$). 
	
	To facilitate that latter, Chung et al.~\cite{CFHL21} introduced a framework that, in certain cases, allows to bound this alternative figure of merit by means of purely classical reasoning. We briefly recall here some of the core elements of this framework, which are relevant to us. Note that \cite{CFHL21} considers the parallel-query model, where in each of the $q$ (sequential) interactions with the RO, an oracle algorithm $\cal A$ can make $k$ queries simultaneously in parallel with each interaction. Here, we consider the (more) standard model of one query per interaction, i.e., setting $k = 1$. On the other hand, we state and prove a slight generalization of Theorem~5.16 in \cite{CFHL21} (when restricted to $k = 1$). 
	
	A subset $\P \subseteq \DB$ is called a {\em database property}. We say that $D \in \DB$ {\em satisfies} $\P$ if $D \in \P$, and the complement of $\P$ is denoted $\neg\P = \DB \setminus \P$. For such a database property $\P$, \cite{CFHL21} defines \smash{$\qQTC{\bfbot}{\P}$} as the square-root of the maximal probability of $D$ satisfying $\P$ when $D$ is obtained by measuring the internal state of the compressed oracle after the interaction with $\cal A$, maximized over all oracle quantum algorithms $\cal A$ with query complexity $q$, i.e., in short
	\begin{equation}\label{eq:QTC}
		\qQTC{\bfbot}{\P} :=  \max_{\cal A}  \sqrt{\Pr[D \in \P] } \, .
	\end{equation} 
	In the context of Lemma~\ref{lemma:zha} for the case ${\cal Z} = \emptyset$\onlyfullversion{ (see Remark~\ref{rem:Z=0})}, we can define the database property $\P^R := \{ D \!\in\! \DB \,|\, \exists\, {\bf x} \!\in\! {\cal X}^\ell \!: ({\bf x},D({\bf x})) \!\in\! R \}$ induced by~$R$, and thus bound 
	\begin{equation}\label{eq:connection}
		p'({\cal A}) \leq \Pr[ ({\bf x},D({\bf x})) \!\in\! R ] \leq  \Pr[D \in \P^R] \leq \qQTC{\bfbot}{\P^R}^2 
	\end{equation}
	for any oracle quantum algorithm $\cal A$ with query complexity $q$. 
	
	Furthermore, Lemma 5.6 in \cite{CFHL21} shows that for any target database property $\P$ and for any sequence 
	$\P_0,\P_1,\ldots,\P_q$ with $\neg \P_0 = \{\bfbot\}$ and $\P_q = \P$, 
	\begin{equation}\label{eq:CapacityChainRule}
		\qQTC{\bfbot}{\P} \leq \sum_{s=1}^q \QTC{\neg\P_{s-1}}{\P_s} \, ,
	\end{equation}
	where, for any database properties $\P$ and $\P'$, the definition of the {\em quantum transition capacity} $\QTC{\P}{\P'}$ is recalled in Definition~\ref{def:QTC}\onlyshortversion{ in Appendix~\ref{app:QTC}}. 
	
	The nice aspect of the framework of~\cite{CFHL21} is that it provides means to manipulate and bound quantum transition capacities using purely classical reasoning, i.e., without the need to understand and work with the definition. Indeed, for instance Theorem~\ref{thm:simple} below, which is a variant of Theorem~5.17 in (the full version of) \cite{CFHL21}, shows how to bound $\QTC{\P}{\P'}$ by means of a certain classical probability; furthermore, to facilitate the application of such theorems, \cite{CFHL21} showed that the quantum transition capacity satisfies several natural manipulation rules, like $\QTC{\P}{\P'} = \QTC{\P'}{\P}$ (i.e., it is symmetric), and 
	\begin{align}
	\begin{split}\label{eq:QTC_props}
		\QTC{\P \cap {\sf Q}}{\P'} &\leq \min \bigl\{ \QTC{\P}{\P'}, \QTC{ {\sf Q}}{\P'} \bigr\} \qquad \text{and} \\[1ex]
		\min \bigl\{ \QTC{\P}{\P'}, \QTC{\P}{ {\sf Q}'} \bigr\}  &\leq \QTC{\P}{\P' \cup {\sf Q}'} \leq \QTC{\P}{\P'} + \QTC{\P}{{\sf Q}'} \, ,
		\end{split}
	\end{align}
	which allow to decompose complicated capacities into simpler ones. 
	Therefore, by means of the above series of inequalities with $p$ from Lemma~\ref{lemma:zha} on the left hand side, it is possible (in certain cases) to bound the success probability of any oracle quantum algorithm $\cal A$ in the QROM by means of the following recipe: (1) Choose suitable transitions $\P_{s-1} \to \P_s$, (2) decompose the capacities $\QTC{\neg\P_{s-1}}{\P_s}$ into simpler ones using manipulation rules as above, and (3) bound the simplified capacities by certain classical probabilities, exploiting results like Theorem~\ref{thm:simple}. We will closely follow this recipe. 
	
	In order to state and later use Theorem~\ref{thm:simple}, we need to introduce the following additional concepts. As explained above, there is no need to actually spell out the definition of the quantum transition capacity in order to use Theorem~\ref{thm:simple}; for completeness, and since it is needed for the proof of Theorem~\ref{thm:simple}, we do provide it \shortorfullversion{in Appendix~\ref{app:QTC} (where we also give the proof of Theorem~\ref{thm:simple})}{below}. 
	
	
	For any database $D \in \DB$ and any $x \in \cal X$,\shortorfullversion{
	$
	D|^x := \{ D[x \!\mapsto\! y] \mid y \in {\cal Y} \cup \{\bot\}  \} 
	$}{
	$$
	D|^x := \{ D[x \!\mapsto\! y] \mid y \in {\cal Y} \cup \{\bot\}  \} 
	$$
	}
	denotes the set of all databases that coincide with $D$ outside of $x$. 
	Furthermore, for a database property~$\P$, 
	$$
	\P|_{D|^x} := \{ y \in {\cal Y} \cup \{\bot\} \mid D[x \!\mapsto\! y] \in  \P \} \subseteq {\cal Y} \cup \{\bot\}
	$$
	denotes the set of values $y$ for which $D[x \!\mapsto\! y]$ satisfies $\P$. 
	Following the convention used in \cite{CFHL21}, we identify the subset $\P|_{D|^x} \subseteq {\cal Y} \cup \{\bot\}$ with the projector $\P|_{D|^x}  = \sum_y \proj{y}$ acting on $\C[{\cal Y} \cup \{\bot\}]$, where the sum is over all $y \in \P|_{D|^x}$. 
	
	\begin{definition}[Def.~5.5 of~\cite{CFHL21}{\rm, case $k=1$}]\label{def:QTC}
		Let $\P, \P'$ be two database properties. Then, the {\em quantum transition capacity} (of order $1$) is defined as
		$$
		\QTC{\P}{\P'} := \max_{\bfx,\hat\bfy,D} \|\P'|_{D|^\bfx} \,\cO_{\bfx, \hat\bfy} \, \P|_{D|^\bfx}\| 
		$$
		where the max is over all $\bfx \in {\cal X}^k$, $\hat\bfy \in \hat{\cal Y}^k$, and $D \in \DB$. 
	\end{definition}
	The following is a variation of Theorem~5.17 in (the full version of) \cite{CFHL21}, obtained by restricting $k$ to $1$. On the other hand, we exploit and include some symmetry that is not explicit in the original statement. The proof\onlyshortversion{, given in Appendix~\ref{app:QTC},} is a small adjustment to the original proof. 
	
	\begin{theorem}\label{thm:simple}
		Let $\P$ and $\P'$ be database properties with trivial intersection, i.e., $\P \cap \P' = \emptyset$, and for every $D \in \DB$ and $x \in {\cal X}$ let
		$$
		\L^{x,D} := \left\{\begin{array}{ll}
			\P|_{D|^x} & \text{if $\bot \in \P'|_{D|^x}$} \\[1ex]
			\P'|_{D|^x} & \text{if $\bot \in \P|_{D|^x}$} \; ,
		\end{array}\right. 
		$$
		with $\L^{x,D}$ being either of the two if $\bot \not\in \P|_{D|^x} \cup \P'|_{D|^x}$.%
		\footnote{By the disjointness requirement, $\bot$ cannot be contained in both. }
		Then
		$$
		\QTC{\P}{\P'} \leq \max_{x,D}\sqrt{10 P\bigl[U \!\in\! \L^{x,D} \bigr]} \, ,
		$$
		where $U$ is uniform over $\cal Y$, and the maximization can be restricted to $D \in \DB$ and $x \in \cal X$ for which both $\P|_{D|^x}$ and $\P'|_{D|^x}$ are non-empty.  
	\end{theorem}
	
	\begin{remark}\label{rem:IndOfD(x)}
		Both, $\P|_{D|^x}$ and $\P'|_{D|^x}$, and thus also $\L^{x,D}$, do not depend on the value of $D(x)$, only on the values of $D$ outside of $x$. 
	\end{remark}
	
	\begin{proof}
		For any $D \in \DB$ and $x \in \cal X$, we observe that 
		$$
		\|\P'|_{D|^x} \,\cO_{x, \hat y} \, \P|_{D|^x}\| = \| \P|_{D|^x} \,\cO_{x , -\hat y} \, \P'|_{D|^x}\| \, ,
		$$ 
		and so it is sufficient to argue for the case when $\L^{x,D}$ is set to $\P'|_{D|^x}$. By the disjointness requirement, as subsets of ${\cal Y} \cup \{\bot\}$, the complement of $\L^{x,D} = \P'|_{D|^x}$ is a superset of $\P|_{D|^x}$. Thus, as projections acting on $\C[{\cal Y} \cup \{\bot\}]$, $\P|_{D|^x} \leq \I - \L^{x,D}$. Therefore, the above norm is upper bounded by $\|\L^{x,D} \,\cO_{x , y} \, (\I - \L^{x,D})\|$. Given that $\bot \not\in \L^{x,D}$, the square norm $\|\L^{x,D} \,\cO_{x , \hat y} \, (\I - \L^{x,D})\|^2$ can be upper bounded exactly as in the proof of Theorem~5.17 in \cite{CFHL21} by $10 P\bigl[U \!\in\! \L^{x,D} \bigr]$, giving the claimed bound. 
		\qed
	\end{proof}

	\subsection{An Improved Variant of Zhandry's Lemma}\label{sec:ImprovedLemma}
	
	We show here an alternative to Zhandry's lemma (Lemma~\ref{lemma:zha}), which offers a better bound in typical applications. To start with, note that Lemma~\ref{lemma:zha} considers an algorithm $\cal A$ that not only outputs ${\bf x} = (x_1,\ldots,x_\ell)$ but also ${\bf y} = (y_1,\ldots,y_\ell)$, where the latter is supposed to be the point-wise hash of~$\bf x$; indeed, this is what is being checked in the definition of the probability $p$, along with $({\bf x},{\bf y},z) \in R$. This requirement is somewhat unnatural, in that an algorithm $\cal A$ for, say, finding a collision, i.e., $x_1 \neq x_2$ with $H(x_1) = H(x_2)$, does {\em not} necessarily output the (supposed to be equal) hashes $y_1 = H(x_1)$ and $y_2 = H(x_2)$. Of course, this is no problem since one can easily transform such an algorithm $\cal A$ that does not output the hashes into one that does, simply by making a few more (classical) queries to the RO at the end of the execution, and then one can apply Lemma~\ref{lemma:zha} to this tweaked algorithm $\tilde{\cal A}$. 
	
	We show below that if we anyway consider this tweaked algorithm $\tilde{\cal A}$, which is {\em promised} to query the RO to obtain and then output the hashes of ${\bf x} = (x_1,\ldots,x_\ell)$, then we can actually improve the bound and avoid the square-roots in Lemma~\ref{lemma:zha}. On top, the proof is much simpler than Zhandry's proof for his lemma. 
	\shortorfullversion{}{
	
}
	At the core is the following lemma; Corollary~\ref{cor:link} below then puts it in a form that is comparable to Lemma~\ref{lemma:zha} and shows the improvement.

	\begin{lemma}\label{lemma:link} 
		Let $\mathcal A$ be an oracle quantum algorithm that outputs ${\bf x} = (x_1,...,x_\ell) \in {\mathcal X}^\ell$ and $z \in \cal Z$. Let $\tilde{\mathcal A}$ be the oracle quantum algorithm that runs $\mathcal A$, makes $\ell$ classical queries on the outputs $x_i$ to obtain ${\bf y} = H(\bf{ x})$, and then outputs  $({\bf x}, {\bf y},z)$. When $\tilde{\mathcal A}$ interacts with the compressed oracle instead, and at the end $D$ is obtained by measuring the internal state of the compressed oracle, then, conditioned on $\tilde{\mathcal A}$'s output $({\bf x}, {\bf y},z)$, 
		\begin{equation*}
			\Pr[ {\bf y} \!=\! D({\bf x}) | ({\bf x}, {\bf y},z)]\ge 1 - \frac{2\ell}{|\mathcal Y|} \, .
		\end{equation*}
	\end{lemma}

	\begin{proof}
		Consider first $\tilde{\mathcal A}$ interacting with the {\em purified} (yet uncompressed) oracle. Conditioned on $\tilde{\mathcal A}$'s output $({\bf x}, {\bf y},z)$, the state of the oracle is then supported by $\ket{H}$ with $H(x_i) = y_i$ for all $i \in \{1,\ldots,\ell\}$, i.e., the registers labeled by $x_1,...,x_\ell$ are in state $\ket{y_1}\cdots\ket{y_\ell}$. Given that the compressed oracle is obtained by applying ${\sf Comp}$ to all the registers, we thus have that
		\begin{align*}
			\Pr[y_i\!=\!y_i' | ({\bf x}, {\bf y},z)] &= \big|\bra{y_i} {\sf Comp} \ket{y_i}\big|^2 = \Big|\bra{y_i} \Bigl(\ket{y_i} + \textstyle\frac{1}{\sqrt{|\mathcal Y|}}(\ket{\bot} - \ket{\hat 0})\Big)\Big|^2 \\
			&= \Big| 1 - {\textstyle\frac{1}{\sqrt{|\mathcal Y|}}}\braket{y_i}{\hat 0} \Big|^2 = \Big| 1 - {\textstyle\frac{1}{|\mathcal Y|}} \Big|^2 \geq 1 - \frac{2}{|\mathcal Y|} \, .
		\end{align*}
		Applying union bound concludes the claim. 
		\qed
	\end{proof}
\onlyshortversion{	In Appendix~\ref{app:Corollaries}, we show a couple of corollaries of Lemma~\ref{lemma:link}, one where $\tilde{\mathcal A}$ may make a more involved computation on $\bf x$, possibly calling $H$ adaptively, and one is put in a form that can be nicely compared with Lemma~\ref{lemma:zha}, understanding that typically Lemma~\ref{lemma:zha} is applied to~$\tilde{\mathcal A}$. }
	The following generalization of Lemma~\ref{lemma:link} follows immediately by enhancing $\cal A$ so that it computes and outputs all the values $x$ that need to be queried in order to compute ${\cal F}^H(z)$, and then apply Lemma~\ref{lemma:link} above. 
	
	\begin{corollary}\label{cor:linkgen} 
		Let $\mathcal A$ be an oracle quantum algorithm that produces an arbitrary output $z \in \cal Z$, and let $\cal F$ be an arbitrary {\em classical} $\ell$-query oracle algorithm. Let $\tilde{\mathcal A}:= {\cal F} \circ {\cal A}$ be the oracle quantum algorithm that first runs $\mathcal A$ to obtain $z$, then $\cal F$ to obtain $y := {\cal F}^H(z)$, and finally outputs $(y,z)$. When $\tilde{\mathcal A}$ interacts with the compressed oracle instead, and at the end $D$ is obtained by measuring the internal state of the compressed oracle, then, conditioned on $\tilde{\mathcal A}$'s output $(y,z)$, 
		\begin{equation*}
			\Pr[ y \!=\! {\cal F}^D(z) | (y,z)]\ge 1 - \frac{2\ell}{|\mathcal Y|}.
		\end{equation*}
	\end{corollary}
	The following corollary of  Lemma~\ref{lemma:link} is put in a form that can be nicely compared with Lemma~\ref{lemma:zha}, understanding that typically Lemma~\ref{lemma:zha} is applied to~$\tilde{\mathcal A}$. 
	
	\begin{corollary}\label{cor:link}
		Let $R \subseteq {\cal X}^\ell \times {\cal Y}^\ell \times {\cal Z}$ be a relation. 
		Let $\cal A$ be an oracle quantum algorithm that outputs ${\bf x} \in {\cal X}^\ell$ and $z \in \cal Z$, and let $\tilde{\mathcal A}$ be as in Lemma~\ref{lemma:link}. Let 
		$$
		p_\circ({\cal A}) := \Pr[ ({\bf x},H({\bf x}),z) \in R ]
		$$ 
		be the considered probability when $\cal A$ has interacted with the RO. Furthermore, let $p(\tilde{\cal A})$, $p'(\tilde{\cal A})$ and $p''(\tilde{\cal A})$ be defined as in Lemma~\ref{lemma:zha} (but now for~$\tilde{\cal A}$). Then 
		$$
		p_\circ({\cal A}) = p(\tilde{\cal A}) \leq p'(\tilde{\cal A}) + \frac{2\ell}{|\mathcal Y|} \, . 
		$$
	\end{corollary}
	For convenience, we recall that 
	\begin{align*}
		p'(\tilde{\cal A}) &=  \Pr[ {\bf y} \!=\! D({\bf x}) \wedge ({\bf x},{\bf y},z) \!\in\! R ] \leq \Pr[ ({\bf x},D({\bf x}),z) \!\in\! R ] \, . 
	\end{align*} 
	
	\begin{proof}
		The equality holds by construction of $\tilde{\cal A}$. For the first inequality, we observe that
		\begin{align*}
			p'(\tilde{\cal A}) &
			=   \Pr[{\bf y} \!=\! D({\bf x}) | ({\bf x},{\bf y},z) \!\in\! R] \Pr[({\bf x},{\bf y},z) \!\in\! R ] \\
			&\geq \big(1- \textstyle\frac{2\ell}{|\mathcal Y|}  \big) \Pr[({\bf x},{\bf y},z) \!\in\! R ] 
			\geq \big(1- \textstyle\frac{2\ell}{|\mathcal Y|}  \big) p(\tilde{\cal A})  
			\geq  p(\tilde{\cal A})  - \textstyle\frac{2\ell}{|\mathcal Y|} \, ,
		\end{align*}
		where the first inequality is by Lemma~\ref{lemma:link}. The second and last inequality in the statement holds trivially by definition of $p'$. 
		\qed
	\end{proof}

	
	\section{Some Background on (Non-)Interactive Proofs}
	
	Throughout this and later sections, we consider a hash function $H: {\cal X} \to {\cal Y}$, to be modeled as a RO then. 
	For concreteness and simplicity, we assume that all relevant variables are encoded as bit strings, and that we can therefore choose $H: \{0,1\}^{\leq B} \to \{0,1\}^n$ for sufficiently large $B$ and $n$.%
	\footnote{$B$ and $n$ may depend on the security parameter $\lambda \in \N$. We will then assume that $B$ and $n$ can be computed from $\lambda$ in polynomial time (in $\lambda$). } 
	
		Let $\{\INST_\lambda\}_{\lambda \in \N}$ and $\{{\cal W}_\lambda\}_{\lambda \in \N}$ be two families of sets, with the members being labeled by the security parameter $\lambda \in \N$. Let $R_\lambda \subseteq \INST_\lambda \times {\cal W}_\lambda$ be a relation that is polynomial-time computable in $\lambda$. $w \in {\cal W}_\lambda$ is called a {\em witness} for $\instance \in \INST_\lambda$ if $R_\lambda(\instance,w)$, and $L_\lambda := \{\instance \in \INST_\lambda \mid \exists \, w \in {\cal W}_\lambda : R_\lambda(\instance,w) \}$. 
		
	Below, we recall some concepts in the context of interactive and non-interactive proofs for such families $\{R_\lambda\}_{\lambda \in \N}$ of relations. We start by discussing the aspired security definition for non-interactive proofs.

	\subsection{Non-interactive Proofs and Online Extractability}
	
	An {\em non-interactive proof in the random-oracle model} for a family $\{R_\lambda\}_{\lambda \in \N}$ of relations consists of a pair $({\cal P},{\cal V})$ of oracle algorithms, referred to as {\em prover} and {\em verifier}, both making queries to the RO $H: {\cal X} \to {\cal Y}$. The prover $\cal P$ takes as input $\lambda \in \N$ and an instance $\instance \in L_\lambda$ and outputs a {\em proof} $\pi \in \Pi_\lambda$, and ${\cal V}$ takes as input $\lambda \in \N$ and a pair $(\instance,\pi) \in \INST_\lambda \times \Pi_\lambda$ and outputs a Boolean value, $0$ or $1$, or $\tt accept$ or $\tt reject$. The verifier ${\cal V}$ is required to run in time polynomial in $\lambda$, while, {\em per-se}, $\cal P$ may have unbounded running time.%
	\footnote{Alternatively, one may consider a witness $w$ for $\instance$ to be given as additional input to $\cal P$, and then ask $\cal P$ to be polynomial-time as well.} 
	
	\shortorfullversion{
	By default, we require correctness and soundness, i.e., that for any $\lambda \in \N$ and any $\instance \in L_\lambda$ the probability 
	$
	\Pr\bigl[{\cal V}^H(\lambda,\instance,\pi) : \pi \leftarrow {\cal P}^H(\lambda,\instance)\bigr] 
	$
	is close to $1$, while for any $\lambda \in \N$ and any oracle quantum algorithm ${\cal P}^*$ with bounded query complexity the probability 
	$
	\Pr\bigl[  \instance \not\in L_\lambda  \, \wedge \, {\cal V}^H(\lambda,\instance,\pi) : (\instance, \pi) \leftarrow {{\cal P}^*}{}^H(\lambda) \bigr] 
	$
	is close to vanishing. 
	}
	{	By default, we require correctness and soundness, i.e., that for any $\lambda \in \N$ and any $\instance \in L_\lambda$
	$$
	\Pr\bigl[{\cal V}^H(\lambda,\instance,\pi) : \pi \leftarrow {\cal P}^H(\lambda,\instance)\bigr] \geq 1-\varepsilon_\text{\rm cor}(\lambda) , 
	$$
	while for any $\lambda \in \N$ and any oracle quantum algorithm ${\cal P}^*$ (a {\em dishonest prover}) with query complexity $q$
	$$
	\Pr\bigl[  \instance \not\in L_\lambda  \, \wedge \, {\cal V}^H(\lambda,\instance,\pi) : (\instance, \pi) \leftarrow {{\cal P}^*}{}^H(\lambda) \bigr] \leq \varepsilon_\text{\rm snd}(\lambda,q,n)
	$$
	for certain $\varepsilon_\text{\rm cor}$ and $\varepsilon_\text{\rm snd}$, respectively referred to as {\em correctness error} and {\em soundness error}. }
	The fact that the instance $\instance$, for which ${\cal P}^*$ tries to forge a proof, is not given as input to ${\cal P}^*$ but is instead chosen by ${\cal P}^*$ is referred to as ${\cal P}^*$ being {\em adaptive}. 
	
	We now move towards defining {\em online extractability} (for adaptive ${\cal P}^*$). For that purpose, let ${\cal P}^*$ be a dishonest prover as above, except that it potentially outputs some additional auxiliary (possibly quantum) output $Z$ next to $(\instance, \pi)$. 
	We then consider an interactive algorithm $\cal E$, called {\em online extractor}, which takes $\lambda \in \N$ 
	as input and simulates the answers to the oracle queries in the execution of ${\cal V}^H \circ {\cal P}^*{}^H(\lambda)$, which we define to run  $(\instance, \pi,Z) \leftarrow {{\cal P}^*}{}^H(\lambda)$ followed by $v \leftarrow {\cal V}^H(\lambda,\instance,\pi)$; furthermore, at the end, $\cal E$ outputs $w \in {\cal W}_\lambda$. 
	We denote the execution of ${\cal V}^H \circ {\cal P}^*{}^H(\lambda)$ with the calls to $H$ simulated by $\cal E$, and considering $\cal E$'s final output $w$ as well, as $(\instance,\pi,Z;v;w) \leftarrow {\cal V}^{\cal E} \circ {\cal P}^*{}^{\cal E}(\lambda)$.
	
	\begin{definition}\label{def:OnlineExtr}
		A non-interactive proof in the (quantum-accessible) RO model (QROM) for $\{R_\lambda\}_{\lambda \in \N}$ is a {\em proof of knowledge with online extractability} (PoK-OE) against adaptive adversaries if there exists an online extractor $\cal E$, and functions $\varepsilon_\text{\rm sim}$ (the {\em simulation error}) and $\varepsilon_\text{\rm ex}$ (the {\em extraction error}), with the following properties. For any $\lambda \in \N$ and for any dishonest prover ${\cal P}^*$ with query complexity~$q$, 
		\shortorfullversion{
        \begin{align*}
		&\delta\bigl( [(\instance,\pi,Z,v)]_{{\cal V}^H \circ {\cal P}^*{}^H(\lambda)} , [(\instance,\pi,Z,v)]_{{\cal V}^{\cal E} \circ {\cal P}^*{}^{\cal E}(\lambda)} \bigr) \leq \varepsilon_\text{\rm sim}(\lambda,q,n) \qquad\text{and} \\[0.5ex]
		\Pr\bigl[& v = {\tt accept} \,\wedge\, (\instance,w) \not\in R: (\instance,\pi,Z;v;w) \leftarrow {\cal V}^{\cal E} \circ {\cal P}^*{}^{\cal E}(\lambda) \bigr] \leq \varepsilon_\text{\rm ex}(\lambda,q,n) \, .
        \end{align*}
		Furthermore, }
		{
		$$
		\delta\bigl( [(\instance,\pi,Z,v)]_{{\cal V}^H \circ {\cal P}^*{}^H(\lambda)} , [(\instance,\pi,Z,v)]_{{\cal V}^{\cal E} \circ {\cal P}^*{}^{\cal E}(\lambda)} \bigr) \leq \varepsilon_\text{\rm sim}(\lambda,q,n)
		$$
		and 
		$$
		\Pr\bigl[ v = {\tt accept} \,\wedge\, (\instance,w) \not\in R: (\instance,\pi,Z;v;w) \leftarrow {\cal V}^{\cal E} \circ {\cal P}^*{}^{\cal E}(\lambda) \bigr] \leq \varepsilon_\text{\rm ex}(\lambda,q,n) \, .
		$$
		Furthermore, }
		the runtime of $\cal E$ is polynomial in $\lambda+q+n$, and $\varepsilon_\text{\rm sim}(\lambda,q,n)$ and $\varepsilon_\text{\rm ex}(\lambda,q,n)$ are negligible in $\lambda$ whenever $q$ and $n$ are polynomial in $\lambda$. 
		\label{def:PoKOnline}\end{definition}
	
	\begin{remark}
		In the classical definition of a proof of knowledge, the extractor $\cal E$ interacts with ${\cal P}^*$ only, and the verifier $\cal V$ is not explicitly involved, but would typically be run by $\cal E$. Here, in the context of online extractability, it is necessary to explicitly go through the verification procedure, which also makes oracle queries, to determine whether a proof is valid, i.e., for the event $v = {\tt accept}$ to be well defined. 
	\end{remark}

	\subsection{(Commit-and-Open) $\Sigma$-Protocols} \label{sec:CandOsigp}
	
	A  {\em \sigp} is a $3$-round public-coin interactive proof $({\cal P},{\cal V})$ for a relation $R_\lambda \subseteq \INST_\lambda \times {\cal W}_\lambda$, indexed by the security parameter. From now on, we leave any dependencies on the security parameter implicit. We therefore simply write $R$ etc. By definition, a \sigp has the following communication pattern. In the first round, $\cal P$ sends a {\em first message} $a$; in the second round, $\cal V$ sends a random {\em challenge} $c \in \cal C$; and in the third round, $\cal P$ sends a {\em response} $z$ (see Fig.~\ref{fig:plainsigp}\onlyshortversion{ in Appendix~\ref{sec:sigpfigures}}). By a slight abuse of notation, we sometimes write ${\cal V}(\instance,a,c,z)$ for the predicate that determines whether $\cal V$ accepts the transcript $(a,c,z)$ on input $\instance$. 
	
		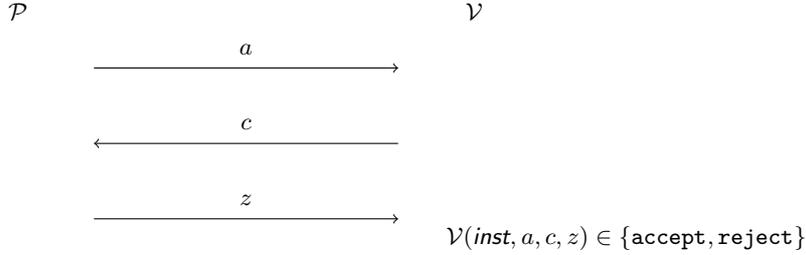
\begin{figure}\centering
		\begin{tikzpicture}
		\begin{pgfonlayer}{nodelayer}
		\node  (0) at (-19, 3) {};
		\node  (1) at (-15, 3) {};
		\node  (2) at (-15, 2) {};
		\node  (3) at (-19, 2) {};
		\node  (4) at (-19, 1) {};
		\node  (5) at (-15, 1) {};
		\node  (7) at (-14, 3.75) {$\cal V$};
		\node  (8) at (-20, 3.75) {$\cal P$};
		\node  (9) at (-17, 3.25) {$a$};
		\node  (10) at (-17, 2.25) {$c$};
		\node  (11) at (-17, 1.25) {$z$};
		\node  (35) at (-13.25, 0.75) {};
		\node  (ar) at (-12, 0.75) {${\cal V}(\instance,a,c,z) \in \{{\tt accept}, {\tt reject}\}$};
		\end{pgfonlayer}
		\begin{pgfonlayer}{edgelayer}
		\draw [style=arrow] (0.center) to (1.center);
		\draw [style=arrow] (2.center) to (3.center);
		\draw [style=arrow] (4.center) to (5.center);
		\end{pgfonlayer}
		\end{tikzpicture}
		\caption{A plain \sigp, formally introduced in Section~\ref{sec:CandOsigp}.\label{fig:plainsigp}}
	\end{figure}

	For the purpose of this work, a {\em commit-and-open \sigp}, or {\em C\&O \sigp} or {\em C\&O protocol} for short, is a \sigp $\Pi=({\cal P},{\cal V})$ of a special form, involving a hash function $H$ that is modeled as a RO.%
	\footnote{One could also refer to \sigps that use non-hash-based commitments, and/or are analyzed in the standard model, as {\em \CnO protocols}, but this is not the scope here. } Concretely, in a C\&O protocol, the transcript $(a,c,z)$ is of the following form (see Fig.~\ref{fig:ordinarysigp}). The first message $a$ consists of  {\em commitments} $y_1,\ldots,y_\ell$, computed as $y_i = H(m_i)$ for {\em messages} $m_1,\ldots,m_\ell \in \cal M$, and possibly an additional string $a_\circ$
	\footnote{Note that $m_i \in \cal M$ may consist of the actual ``message'' (computed by the prover using the witness $w$), possibly concatenated with randomness.}.
	The challenge $c$ is picked uniformly at random from the challenge space ${\cal C} \subseteq 2^{[\ell]}$, which is set to be a subset of $2^{[\ell]}$. Finally, the response $z$ is given by ${\bf m}_c = (m_i)_{i \in c}$. Eventually, $\cal V$ accepts if and only if $H(m_i)=y_i$ for all $i \in c$ and some given predicate $V(\instance,c,{\bf m}_c,a_\circ)$ is satisfied. 
	
	For the above to be meaningful, we obviously need that ${\cal M} \subseteq {\cal X}$, i.e., the bit size of the possible $m_i$'s are upper bounded by $B$. Furthermore, the parameter $n$ determines the hardness of finding a collision in $H$ (in the random oracle model), and thus the level of binding the commitments provide. 

	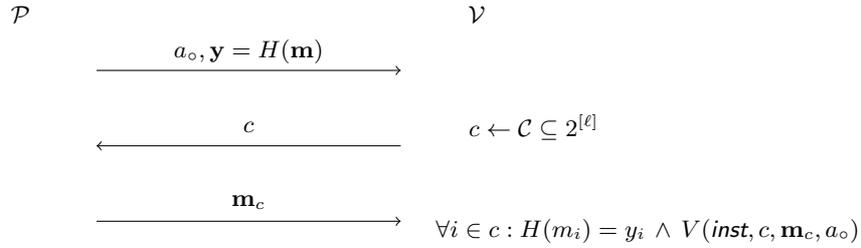
\begin{figure}\centering
		\begin{tikzpicture}
		\begin{pgfonlayer}{nodelayer}
		\node  (0) at (-19, 3) {};
		\node  (1) at (-15, 3) {};
		\node  (2) at (-15, 2) {};
		\node  (3) at (-19, 2) {};
		\node  (4) at (-19, 1) {};
		\node  (5) at (-15, 1) {};
		\node  (7) at (-14, 3.75) {$\cal V$};
		\node  (8) at (-20, 3.75) {$\cal P$};
		\node  (9) at (-17, 3.25) {$a_\circ, {\bf y} = H({\bf m})$};
		\node  (10) at (-17, 2.25) {$c$};
		\node  (11) at (-17, 1.25) {${\bf m}_c$};
		\node  (35) at (-13.25, 0.75) {};
		\node  (in) at (-11.75, 0.875) {$\forall i \in c : H(m_i) = y_i \,\wedge\, V(\instance,c,{\bf m}_c,a_\circ)$};
		\node  (40) at (-13.25, 2.25) {$c \leftarrow{\cal C} \subseteq 2^{[\ell]}$};
		\end{pgfonlayer}
		\begin{pgfonlayer}{edgelayer}
		\draw [style=arrow] (0.center) to (1.center);
		\draw [style=arrow] (2.center) to (3.center);
		\draw [style=arrow] (4.center) to (5.center);
		\end{pgfonlayer}
		\end{tikzpicture}
		\caption{An (ordinary) C\&O \textSigma-protocol,  formally introduced in Section~\ref{sec:CandOsigp}. 
		\label{fig:ordinarysigp}}
	\end{figure}

	\begin{remark}\label{rem:generalCandO}
		Looking ahead, we may also consider a generalization of the above notion of a C\&O protocol, where the first message is parsed as a {\em single} commitment $y$ of the $\ell$ messages $m_1,\ldots,m_\ell$ and where this commitment is computed by means of an arbitrary ``multi-message'' commitment scheme involving $H$, which has the property that any subset of $m_1,\ldots,m_\ell$ can be opened without revealing the remaining $m_i$'s.  The above component-wise hashing is then one particular instantiation, but alternatively one can for instance also compute $y$ by means of a Merkle tree (see Section~\ref{sec:Octopus}), and then open individual $m_i$'s by revealing the corresponding authentication paths. 
		We stress that the concepts discussed below: the notions of $\frak{S}$-soundness and $\frak{S}$-soundness$^*$ and the probability $p^{\mathfrak S}_{triv}$, do not depend on the choice of commitment scheme, and thus remain unaffected when considering such a {\em \MCnO} protocol. To emphasize the default choice of the commitment scheme, which is element-wise hashing, we sometimes also speak of an {\em ordinary} \CnO protocol. 
	\end{remark}

	\subsection{$\frak{S}$-soundness of \CnO $\Sigma$-Protocols}
	
	We briefly recall the notion of $\frak{S}$-soundness and  $\frak{S}$-soundness$^*$ for \CnO protocols, as considered in~\cite{DFMS21}, which offers a convenient general notion of special soundness, or more generally $k$-soundness for \CnO protocols.  
	A similar notion of $\frak{S}$-soundness naturally exists for plain \sigps, i.e., \sigps in the plain model. For completeness, we formalize the latter in Appendix~\ref{sec:S-soundOrd}. 
	
	Here and below, given a \CnO protocol $\Pi$ with challenge space ${\cal C} \subseteq 2^{[\ell]}$,  we let $\frak{S} \subseteq 2^{\cal C}$ be an arbitrary non-empty, monotone increasing set of subsets $S \subseteq \cal C$, where the monotonicity means that $S \in \frak{S}\, \wedge \,S \subseteq S' \: \Rightarrow \: S' \in \frak{S}$. We then also set $\frak{S}_{\min} := \{S \in \frak{S} \mid S_\circ \subsetneq S \Rightarrow S_\circ \not\in \frak{S} \}$ to be the minimal sets in $\frak{S}$.

	For simplicity, the reader can consider $\frak{S} = \frak{T}_k := \{ S \subseteq {\cal C} \mid |S| \geq k \}$ for some threshold $k$, and thus $\frak{S}_{\min} = \{ S \subseteq {\cal C}  \mid |S| = k \}$. This then corresponds to the notion of $k$-soundness for \CnO protocols, which in turn means that the witness can be computed from valid responses to $k$ (or more) distinct challenges for a given first message $y_1,\ldots,y_\ell$, assuming the messages $m_1,\ldots,m_\ell$ to be uniquely determined by their commitments.

	\begin{definition}[\cite{DFMS21}{\rm\ Def.~5.1}]\label{def:S-sound}
		A \CnO protocol $\Pi$ is $\frak{S}$\emph{-sound} if there exists an efficient deterministic algorithm $\mathcal E_\frak{S}(\instance,m_1,\ldots,m_\ell, a_\circ,S)$ that takes as input an instance $\instance \in \INST$, messages $m_1,\ldots,m_\ell \in {\cal M} \cup \{\bot\}$, a string $a_\circ$, and a set $S \in \frak{S}_{\min}$, and outputs a witness for $\instance$ if $V(\instance, c,{\bf m}_c,a_\circ)$  for all $c \in S$.%
		\footnote{The restriction for $S$ to be in $\frak{S}_{\min}$, rather than in $\frak{S}$, is to avoid an exponentially sized input while asking $\mathcal E_\frak{S}$ to be efficient. }
	\end{definition}
	\onlyfullversion{
	We note that Wikstr\"om~\cite{Wik18} also considers a general notion of special soundness (but then for multi-round protocols); however, the notion in~\cite{Wik18} is more restrictive in that it requires some matroid structure on top. For instance, the $r$-fold parallel repetition of a $k$-sound protocol does not fit into the formalism by Wikstr\"om. 
	}
	
	A slightly stronger condition than $\frak{S}$\emph{-soundness} is the following variant, which differs in that the extractor needs to work as soon as there {\em exists} a set $S$ as specified, without the extractor being given $S$ as input. We refer to~\cite{DFMS21} for a more detailed discussion of this aspect. As explained there, whether $S$ is given or not often makes no (big) difference. 
	
	For instance, when $\frak{S}_{\min}$ consists of a polynomial number of sets $S$ then the extractor can do a brute-force search to find~$S$, and so $\frak{S}$-soundness$^*$ is then implied by $\frak{S}$-soundness. Also, the $r$-fold parallel repetition of a $\frak{S}$-sound protocol, which by default is a $\frak{S}^{\vee r}$-sound protocol (see~\cite{DFMS21}), is automatically $\frak{S}^\vee$-sound$^*$ if $\frak{S}_{\min}$ is polynomial in size:  the extractor can then do a brute-force search in every repeated instance. 
	
	\begin{definition}[\cite{DFMS21}{\rm\ Def.~5.2}]\label{def:S-sound-star}
		A \CnO protocol $\Pi$ is $\frak{S}$\emph{-sound}$^*$ if there exists an efficient deterministic algorithm $\mathcal E^*_{\frak{S}}(\instance,m_1,\ldots,m_\ell, a_\circ)$ that takes as input an instance $\instance \in \INST$ and strings $m_1,\ldots,m_\ell \in {\cal M} \cup \{\bot\}$ and $a_\circ$, and it outputs a witness for $\instance$ if there exists $S \in \frak{S}$ such that  $V(\instance, c,{\bf m}_c,a_\circ)$ for all $c \in S$. 
	\end{definition}
	As in~\cite{DFMS21}, we define
	\begin{align} \label{eq:ptriv}
		p^{\mathfrak S}_{triv} := \frac{1}{|\mathcal{C}|} \max_{\hat S \not\in  \frak{S}} |\hat S| \, , 
	\end{align}
	capturing the ``trivial'' attack of picking a set $\hat S = \{\hat c_1,\ldots,\hat c_{m}\} \not\in \frak{S}$ of challenges $\hat c_i \in \mathcal{C}$ and then prepare $\hat{\bf m} = (\hat m_1,\ldots,\hat m_\ell)$ and $a_\circ$ in such a way that $V(\instance, c, \hat{\bf m}_c,a_\circ)$ holds if $c \in \hat S$. After committing to $\hat m_1,\ldots,\hat m_\ell$, the prover can successfully answer to challenges $c \in \hat S$. 
	
	\subsection{The Fiat-Shamir \Trafo of (\CnO) $\Sigma$-Protocols}
	
	The Fiat-Shamir (FS) \trafo \cite{Fiat1987} turns arbitrary $\Sigma$-protocols into non-interactive proofs in the random oracle model by setting the challenge $c \in \cal C$ to be the hash of the instance and the first message $a$. For this \trafo to work smoothly, it is typically assumed that $|{\cal C}|$ is a power of $2$ and its elements are represented as bit strings of size $\log|{\cal C}|$, so that one can indeed set $c$ to be (the first $\log|{\cal C}|$ bits of) the hash $H(\instance,a)$. The assumption on $|{\cal C}|$ is essentially without loss of generality (WLOG), since one can always reduce the size of $|{\cal C}|$ to the next lower power of $2$, at the cost of losing at most $1$ bit of security. 
	However, for a \CnO $\Sigma$-protocol, where a challenge space $\cal C$ is a (typically strict) subset of $2^{[\ell]}$, there is not necessarily a natural way to represent $c \in \cal C$ as a bitstring of size $\log|{\cal C}|$. 
	Therefore, we will make it explicit that the challenge-set $c \in {\cal C}  \subset 2^{[\ell]}$ is computed from the ``raw randomness'' $H(\instance, y_1,\ldots,y_\ell, a_\circ)$ in a deterministic way as  $c = \gamma \circ H(\instance, y_1,\ldots,y_\ell, a_\circ)$ for an appropriate function $\gamma: {\cal Y} \to {\cal C}$, mapping a uniformly random hash in $\cal Y$ to a random challenge-set in $\cal C$. 
	Obviously, for $H(\instance, y_1,\ldots,y_\ell, a_\circ)$ to be defined, in addition to ${\cal M} \subseteq {\cal X}$ we also need that $\INST \times {\cal Y}^\ell \subseteq \cal X$, which again just means that $B$ needs to be large enough. We write ${\sf FS}[\Pi]$ for the \FST of a (\CnO) $\Sigma$-protocol $\Pi$. 
	
	\begin{remark}
		Additionally, we need that $n$ is sufficiently large, so that there is a sufficient amount of randomness in the hash value $H(\instance, y_1,\ldots,y_\ell)$ in order to be mapped to a random $c \in \cal C$. The canonical choice for $\gamma$ is then the function that the {\em interactive} verifier applies to his local randomness to compute the random challenge $c \in {\cal C}$. To simplify the exposition, we assume that $n$ is indeed sufficiently large. Otherwise, one can simply set ${\cal Y} := \{0,1\}^{n'}$ instead, for sufficiently large $n'$, and then let $y_i$ be $H(m_i)$ {\em truncated} to the original number $n$ of bits again. This truncation has no effect on our results. 
	\end{remark}
	
	\begin{remark}
		We assume WLOG that the two kinds of inputs to $H$, i.e., $m_i$ and $(\instance, y_1,\ldots,y_\ell, a_\circ)$, are differently formatted, e.g., bit strings of different respective sizes or prefixes (this is referred to as \emph{domain separation}). In other words, we assume that ${\cal M}$ and $\INST \times {\cal Y}^\ell$ are disjoint. 
	\end{remark}
	
	\begin{remark}\label{rem:NoString}
		When considering the adaptive security of a \FST ${\sf FS}[\Pi]$ of a \CnO protocol $\Pi$ for a relation $R$, the additional string $a_\circ$, which may be part of the first message $a$ of the original protocol $\Pi$, may WLOG be considered to be part of the instance $\instance$ instead. 
		
		Indeed, any dishonest prover ${\cal P}^*$ against ${\sf FS}[\Pi]$, which (by Definition~\ref{def:OnlineExtr}) outputs an instance $\instance$ and a proof $\pi = (a_\circ, y_1,\ldots y_\ell)$, can alternatively be parsed as a dishonest prover that outputs an instance $\instance' = (\instance, a_\circ)$ and a proof $\pi' = (y_1,\ldots y_\ell)$. Thus, ${\cal P}^*$ can be parsed as a dishonest prover against ${\sf FS}[\Pi']$, where the \CnO protocol $\Pi'$ works as $\Pi$, except that $a_\circ$ is considered as part of the instance, rather than as part of the first message, and thus $\Pi'$ is a \CnO protocol for the relation $((\instance,a_\circ),w)\in R' :\Leftrightarrow  (\instance,w)\in R$.%
		\footnote{We do not specify the local computation of the honest prover ${\cal P}'$ in $\Pi' = ({\cal P}',{\cal V}')$, i.e., how to act when $a_\circ$ is part of the input, and in general it might not be efficient, but this is fine since we are interested in the security against dishonest provers. }
		Therefore, security (in the sense of Definition~\ref{def:OnlineExtr}) for ${\sf FS}[\Pi']$ implies that of ${\sf FS}[\Pi]$. 
	\end{remark}

	\section{Online Extractability of the FS-\Trafo: \\ The Case of Ordinary \CnO Protocols}\label{sec:FSCnO}
	
	We now consider the \FST ${\sf FS}[\Pi]$ of an ordinary \CnO protocol $\Pi$. Our goal is to show that ${\sf FS}[\Pi]$ admits online extraction. We note that by exploiting Remark~\ref{rem:NoString}, we may assume WLOG that the first message of $\Pi$ consists of the commitments $y_1,\ldots,y_\ell$ only, and no additional string~$a_\circ$. In Section~\ref{sec:MCnO}, we then consider the case of \MCnO protocols. 
	
	Our analysis \onlyfullversion{of the online extractability }of ${\sf FS}[\Pi]$ uses the framework of Chung et al.~\cite{CFHL21}, discussed and outlined in Section~\ref{sec:Prelim}. Thus, at the core of our analysis is a bound on a certain quantum transition capacity. This is treated in the upcoming subsection. 
	
	\subsection{Technical Preface}
	
	
	We first introduce a couple of elementary database properties (related to CoLlisions and the SiZe of the database) that will be useful for us: 
	$$
	\CL := \{D \,|\, \exists\, x \!\neq\! x':D(x) \!=\! D(x') \!\neq\! \bot\} 
	\;\text{ and }\;
	\SZ := \{D \,|\, \#\{z|D(z) \!\neq\! \bot\}\leq s \} .
	$$ 
	Next, for an instance $\instance \in \INST$, we want to specify the database property that captures a cheating prover that succeeds in producing an accepting proof while fooling the extractor. For the purpose of specifying this database property, we introduce the following notation. 
	For a given database $D \in \DB$ and for a commitment $y \in \cal Y$, we define $D^{-1}(y)$ to be the smallest $x \in \cal X$ with $D(x) = y$, with the convention that $D^{-1}(y) := \bot$ if there is no such $x$, as well as $D^{-1}(\bot):=\bot$. \shortorfullversion{B}{We note that b}y removing collisions, we ensure that there is at most one such $x$; thus, taking the smallest one in case of multiple choices is not important but only for well-definedness. \shortorfullversion{}{

}The database property of interest can now be defined as
	\begin{equation} \label{eq:defSUC}
		\SUC := \!\left\{D \,\bigg|\, \begin{array}{c}\exists\, {\bf y} \in {\cal Y}^\ell \text{ and }\instance\in \INST\text{ so that } {\bf m}:= D^{-1}({\bf y}) \text{ satisfies} \\
			V(\instance,c,{\bf m}_c) \text{ for } c := \gamma \circ D(\instance, {\bf y})  \,\;\text{and}\;\,
			\big(inst,{\cal E}^*(\instance,  {\bf m})\big) \not\in R
		\end{array} \right\}\! .
	\end{equation}
	
	Informally, assuming no collisions (i.e., restricting to $D \not\in \CL$), the database property $\SUC$ captures whether a database $D$ admits a {\em valid} proof $\pi = ({\bf y},{\bf m}_c)$ for an instance $\instance$ for which the (canonical) extractor, which first computes $\bf m$ by inverting $D$ and then runs ${\cal E}^*$, {\em fails} to produce a witness. 
	
	Our (first) goal is to show that $\qQTC{\bot}{\SUC \cup \CL}$ is small, capturing that it is unlikely that after $q$ queries the compressed database contains collisions or admits a valid proof upon which the extractor fails. Indeed, we show the following, where $p^{\mathfrak S}_{triv}$ is the trivial cheating probability of $\Pi$ as defined in~\eqref{eq:ptriv}. 
	
	\begin{lemma}\label{lem:SUCCLbound}
		$
		\qQTC{\bot}{\SUC \cup \CL} \leq 2eq^{3/2} 2^{-n/2} + q \sqrt{10 \max\left(q \ell\cdot 2^{-n}, p_{triv}^{\mathfrak S}\right)}
		$.
	\end{lemma}
	\shortorfullversion{We begin with an outline of the proof.}{The formal proof is given below; we first give some informal outline here.} In a first step, by using (\ref{eq:CapacityChainRule}) and union-bound-like properties of the transition capacity, and additionally exploiting a bound from~\cite{CFHL21} to control the transition capacity of $\CL$, we reduce the problem to bounding the quantum transition capacity $ \QTC{ \SZ \backslash \SUC}{\SUC}$ for $s<q$. 
	Informally, this capacity is a measure of the ``likelihood''\,---\,but then in a {\em quantum}-sense\,---\,that a database $D \in \DB$ that is bounded in size and not in $\SUC$ 
	turns into a database $D'$ that {\em is} in $\SUC$, when $D$ is updated to $D' = D[x\!\mapsto\!U]$ with $U$ uniformly random in~$\cal Y$. 
	
	We emphasize that\onlyfullversion{ in the considered quantum setting,} the state of the compressed oracle at any point is a {\em superposition} of databases, and a query is made up of a {\em superposition} of inputs; nevertheless, due to Theorem~\ref{thm:simple}, the above classical intuition is actually very close to what needs to be shown to rigorously bound the considered quantum transition capacity. Formally, as will become clear in the proof below, we need to show that for any database $D \in \SZ \backslash \SUC$ and for any $x \in \cal X$ with $D(x) = \bot$, the probability that $D[x\!\mapsto\!U] \in \SUC$ is small. Below, this probability is bounded in the {\em Case~2} and {\em Case~3} parts of the proof, where the two cases distinguish between~$x$ being a ``commit query'' or a ``challenge query''. 
	
	Informally, for $D$ with $D(x) = \bot$, if $x$ is a ``commit query'' then assigning a value to $D(x)$ can only \onlyfullversion{make a difference, i.e.,} turn $D \not\in \SUC$ into $D[x\!\mapsto\!u] \in \SUC$, if $u$ is a coordinate of some ${\bf y} \in {\cal Y}^\ell$ for which $D(\instance, {\bf y}) \neq \bot$ for some $\instance$. Indeed, otherwise, $D[x\!\mapsto\!u]$ does not contribute to a valid proof $\pi$ that did not exist before. Thus,  given the bound $s < q$ on the size of $D$, this happens with probability at most $q\ell/2^n$ for a random $u$. Similarly, if $x$ is a ``challenge query'', i.e. of the form $x = (\instance,{\bf y})$, then assigning a value $u$ to $D(x)$ can only make a difference if $V(\instance,c, {\bf m}_c)$ is satisfied for $c = \gamma(u)$ and ${\bf m} = D^{-1}({\bf y})$, while ${\cal E}^*(\instance,  {\bf m})$ is not a witness for $\instance$. However, for a random $u$, this is bounded by $p_{triv}^{\mathfrak S}$. 
	
	But then, on top of the above, due to the quantum nature of the quantum transition capacity,%
	\footnote{At the core, this is related to the reversibility of quantum computing and the resulting ability to ``uncompute'' a query. }
	Theorem~\ref{thm:simple} requires to also show the ``reverse'', i.e., that for any $D \in \SUC$ and for any $x \in \cal X$ with $D(x) \neq \bot$, the probability that $D[x\!\mapsto\!U] \in  \SZ \backslash \SUC$ is small; this is analyzed in {\em Case~1} below. 
	
	Thus, by exploiting the framework of~\cite{CFHL21}, the core of the reasoning is purely classical, very closely mimicking how one would have to reason the classical setting with a classical RO. Due to the rather complex definition of $\SUC$, the formal argument in each case is still somewhat cumbersome.

	\begin{proof}
		We first observe that, by (\ref{eq:CapacityChainRule}) (which is Lemma~5.6 in~\cite{CFHL21}) and basic properties of the quantum transition capacity as in (\ref{eq:QTC_props}),
	\shortorfullversion{
	\begin{align}
		&\qQTC{\bot}{\SUC \cup \CL} 
		\leq  \sum_{s=0}^{q-1} \QTC{\SZ \backslash \SUC \backslash \CL}{\SUC \cup \CL  \cup \neg\SZ[s+1]} \nonumber\\[-1ex]
		&\leq  \sum_{s=0}^{q-1} \big(\QTC{\SZ}{\neg\SZ[s+1]} + \QTC{\SZ \backslash \CL}{\CL} +  \QTC{ \SZ \backslash \SUC}{\SUC} \big) \, .\label{eq:splitthecapacity}
	\end{align}
}{
		\begin{align}
			\qQTC{\bot&}{\SUC \cup \CL} 
			\leq  \sum_{s=0}^{q-1} \QTC{\SZ \backslash \SUC \backslash \CL}{\SUC \cup \CL  \cup \neg\SZ[s+1]} \nonumber\\[-1ex]
			\begin{split}&\leq  \sum_{s=0}^{q-1} \big( \QTC{ \SZ \backslash \SUC \backslash \CL}{\neg\SZ[s+1]}  + \QTC{ \SZ \backslash \SUC \backslash \CL}{\CL } \\
				& \qquad \qquad + \QTC{ \SZ \backslash \SUC \backslash \CL}{\SUC} \big) 
			\end{split}\nonumber\\
			&\leq  \sum_{s=0}^{q-1} \big(\QTC{\SZ}{\neg\SZ[s+1]} + \QTC{\SZ \backslash \CL}{\CL} +  \QTC{ \SZ \backslash \SUC}{\SUC} \big) \, .\label{eq:splitthecapacity}
		\end{align}
	}
		The first term, $\QTC{\SZ}{\neg\SZ[s+1]}$, vanishes, while the second term was shown to be bounded as 
		\begin{equation}\label{eq:collision}
			\QTC{\SZ \backslash \CL}{\CL} \leq 2e\sqrt{(s+1)/|{\cal Y}|} \leq 2e\sqrt{q/2^n}
		\end{equation}
		in Example~5.28 in~\cite{CFHL21}.  Thus, it remains to control the third term, which we will do by means of Theorem~\ref{thm:simple} with $\P := \SZ \setminus \SUC$ and $\P' := \SUC$. 
		
		To this end, we consider arbitrary but fixed $D \in \DB$ and input $x \in \cal X$. By Remark~\ref{rem:IndOfD(x)}, we may assume that $D(x) = \bot$. Furthermore, for $\P|_{D|^x}$ to be non-empty, it must be that $D \in \SZ$, i.e., $D$ is bounded in size. We now distinguish between the following cases for the considered $D$ and $x$. 
		
		\paragraph{Case 1:} $D \in \SUC$. In particular, $\bot \in \SUC|_{D|^x} = \P'_{D|^x}$. So, Theorem~\ref{thm:simple} instructs us to set $\L:= \P_{D|^x}$, where we leave the dependency of $\L$ on $D$ and $x$ implicit to simplify notation. Given that $D \in \SUC$, we can consider $\instance$ and $\bf y$ as promised by the definition of $\SUC$ in~\eqref{eq:defSUC}, i.e., such that
		$V(\instance,c, {\bf m}_{c})$ and $\big(\instance,{\cal E}^*(\instance,  {\bf m})\big) \not\in R$
		for
		$$
		c := \gamma \circ D(\instance, {\bf y})  \quad\text{and}\quad m_i: = D^{-1}(y_i) \, ,
		$$
		where it is understood that ${\bf m} = (m_1,\ldots,m_\ell)$. 
		Recall that $D(x) = \bot$; thus, by definition of the $m_i$'s, it must be that $x \neq m_i$ for all $i$, and the fact that $V(\instance,c, {\bf m}_{c})$ is satisfied for $c$ as defined implies that $x \neq (\instance, {\bf y})$. 
		Furthermore, 
		$$
		u \in \L \:\Longleftrightarrow\: D[x \!\mapsto\! u] \in \P \:\Longrightarrow\: D[x \!\mapsto\! u] \not\in \SUC  \:\Longrightarrow\: u \in \{y_1,\ldots,y_\ell\} \, ,
		$$
		where the last implication is easiest seen by contraposition: Assume that $u \not\in \{y_1,\ldots,y_\ell\}$. Then, also recalling that $x \neq m_i$, we have that $m_i = D^{-1}(y_i) = D[x \!\mapsto\! u]^{-1}(y_i)$. But also $c = \gamma \circ D(\instance, {\bf y}) = \gamma \circ D[x \!\mapsto\! u](\instance, {\bf y})$. Together, this implies that the defining property of $\SUC$ is also satisfied for $D[x \!\mapsto\! u]$, i.e., $D[x \!\mapsto\! u] \in \SUC$, as was to be shown. 
		Thus, we can bound
		\begin{equation}\label{eq:case1}
			P[U \!\in\! \L] \leq P[ U \!\in\! \{y_1,\ldots,y_\ell\}] \leq \frac{\ell}{|{\cal Y}|} \, .
		\end{equation}
		
		\paragraph{Case 2:}  $D \not\in \SUC$, and $x$ is a ``commit query'', i.e., $x=m \in \cal{M}$. In particular, $\bot \not\in \P'|_{D|^x}$ (by the assumption that $D(x)=\bot$) and so in light of Theorem~\ref{thm:simple} we may choose $\L := \P'|_{D|^x}$. We then have
		\begin{equation}\label{eq:needstochangesomething}
			u \in \L \:\Longleftrightarrow\: D[x \!\mapsto\! u] \in \P' = \SUC  \:\Longrightarrow\:  \exists \, \instance,{\bf y}, i:  D(\instance, {\bf y}) \neq \bot \,\wedge\, u = y_i \, .
		\end{equation}
		\shortorfullversion{The last}{This final} implication can be seen as follows. 
		By definition of $\SUC$, the assumption $D[x \!\mapsto\! u] \in \SUC$ implies the existence of $\instance$ and ${\bf y} = (y_1,\ldots,y_\ell)$ with
		$V(\instance,c, {\bf m}_{c})$ and $\big(\instance,{\cal E}^*(\instance,  {\bf m})\big) \not\in R$
		for
		$$
		c := \gamma \circ D[x \!\mapsto\! u](\instance, {\bf y}) = \gamma \circ D(\instance, {\bf y})  \quad\text{and}\quad m_i: = D[x \!\mapsto\! u]^{-1}(y_i) \, ,
		$$
		where the equality in the definition of $c$ exploits that $x$ is not a ``challenge'' query. 
		With the goal to reach a contradiction, assume that $u \neq y_i$ for all~$i$. This assumption implies that $D[x \!\mapsto\! u](x) = u \neq y_i$. But also $D(x) = \bot \neq y_i$, and hence for all $\xi \in \cal X$ and $i \in \{1,\ldots,\ell\}$: $D(\xi) = y_i \,\Leftrightarrow\, D[x \!\mapsto\! u](\xi) = y_i$. 
		Therefore, $m_i = D[x \!\mapsto\! u]^{-1}(y_i) = D^{-1}(y_i)$ for all~$i$, and the above then implies that $D \in \SUC$, a contradiction. 
		Thus, there exists $i$ for which $u = y_i$; furthermore, $D(\instance, {\bf y}) \neq \bot$ given that $V(\instance,u, {\bf m}_{c})$ is satisfied for $c = \gamma \circ D(\instance, {\bf y}) $. This shows the claimed implication. 
		\shortorfullversion{}{
		
	}
		Thus, we can bound
	\shortorfullversion{
	\begin{equation}\label{eq:case2}
		P[U \!\in\! \L] \leq P[\, \exists \, \instance,{\bf y}, i :   D(\instance, {\bf y}) \neq \bot \,\wedge\, u = y_i ] \leq s \ell/|{\cal Y}| \leq q \ell/|{\cal Y}| \, .
	\end{equation}
}{
	\begin{equation}\label{eq:case2}
			P[U \!\in\! \L] \leq P[\, \exists \, \instance,{\bf y}, i :   D(\instance, {\bf y}) \neq \bot \,\wedge\, u = y_i ] \leq \frac{s \ell}{|{\cal Y}|} \leq \frac{q \ell}{|{\cal Y}|} \, .
		\end{equation}
	}
		\paragraph{Case 3:}  $D \not\in \SUC$, and $x$ is a ``challenge query'', i.e., $x = (\instance, {\bf y}) \in \INST \times {\cal Y}^\ell$. Set ${\bf m}  = (m_1,\ldots,m_\ell)$ for $m_i := D^{-1}(y_i)$. Again, we have that $\bot \not\in \SUC|_{D|^x} = \P'_{D|^x}$, and so by Theorem~\ref{thm:simple} we may set $\L:= \P'_{D|^x}$. 
		Here, we can argue that 
		\shortorfullversion{
		\begin{align*}
			u \in \L\Longleftrightarrow D[x \!\mapsto\! u] \in \P' =\SUC 
			\Longrightarrow V(\instance,u, {\bf m} _{\gamma(u)})\shortorfullversion{\wedge}{\text{ and }}\big(\instance,{\cal E}^*(\instance,  {\bf m})\big) \not\in R \, ,
		\end{align*}
	}{
	\begin{align*}
			u \in \L \:\Longleftrightarrow\: &D[x \!\mapsto\! u] \in \P' =\SUC \\
			\Longrightarrow\:  &V(\instance,u, {\bf m} _{\gamma(u)})\text{ and }\big(\instance,{\cal E}^*(\instance,  {\bf m})\big) \not\in R \, ,
		\end{align*}
	}
		where the final implication can be seen as follows. 
		By definition of $\SUC$, the assumption $D[x \!\mapsto\! u] \in \SUC$ implies the existence of $\instance'$ and ${\bf y}' = (y'_1,\ldots,y'_\ell)$ with
		$V(\instance',u, {\bf m}'_{c})$ and ${\cal E}^*(\instance', {\bf m}')  \neq w$
		for
		$$
		c := \gamma \circ D[x \!\mapsto\! u](\instance', {\bf y'})  \quad\text{and}\quad m'_i: = D[x \!\mapsto\! u]^{-1}(y'_i) = D^{-1}(y'_i) \, ,
		$$
		where the very last equality exploits that $x$ is not a ``commit'' query. 
		With the goal to come to a contradiction, assume that $(\instance',{\bf y'}) \neq (\instance, {\bf y}) = x$. Then, $c  = \gamma \circ D[x \!\mapsto\! u](\instance', {\bf y'}) = \gamma \circ D(\instance', {\bf y'})$, and the above then implies that $D \in \SUC$, a contradiction. Thus, $(\instance',{\bf y'}) = (\instance, {\bf y}) = x$. In particular, ${\bf m}' = {\bf m}$ and $c = \gamma \circ D[x \!\mapsto\! u](\instance', {\bf y'}) = \gamma \circ D[x \!\mapsto\! u](x) = \gamma(u)$. Hence, the claimed implication holds. 
		
		Thus, we can bound 
		\shortorfullversion{
			\begin{align}
			&P[U \!\in\! \L] \leq P[V(\instance,\gamma(U), {\bf m}_{\gamma(U)}) \, \wedge \, {\cal E}^*(\instance,  {\bf m} ) \neq w] \nonumber\\[0.5ex]
			&\leq P[V(\instance,\gamma(U), {\bf m}_{\gamma(U)}) \, \wedge \, S := \{ c \,|\, V(\instance, c, {\bf m}_{c}) \} \not\in \frak{S} ] \nonumber\\[0.5ex]
			&\leq P[\gamma(U) \in S := \{ c \,|\, V(\instance, c, {\bf m}_{c}) \} \not\in \frak{S} ] \leq \max_{S  \not\in \frak{S}}P[\gamma(U) \in S ] \leq p_{triv}^{\mathfrak S} \, .\label{eq:case3}
		\end{align}
	}{
\begin{align}
	P[U \!\in\! \L] &\leq P[V(\instance,\gamma(U), {\bf m}_{\gamma(U)}) \, \wedge \, {\cal E}^*(\instance,  {\bf m} ) \neq w] \nonumber\\[0.5ex]
	&\leq P[V(\instance,\gamma(U), {\bf m}_{\gamma(U)}) \, \wedge \, S := \{ c \,|\, V(\instance, c, {\bf m}_{c}) \} \not\in \frak{S} ] \nonumber\\[0.5ex]
	&\leq P[\gamma(U) \in S := \{ c \,|\, V(\instance, c, {\bf m}_{c}) \} \not\in \frak{S} ] \nonumber\\[0.5ex]
	&\leq \max_{S  \not\in \frak{S}}P[\gamma(U) \in S ] \nonumber\\[-0.7ex]
	&\leq p_{triv}^{\mathfrak S} \, .\label{eq:case3}
\end{align}
}
		By Theorem \ref{thm:simple}, we now get
	\shortorfullversion{
			\begin{align}
		&	\QTC{ \SZ \backslash \SUC \backslash \CL}{\SUC}\leq \max_{x,D}\sqrt{10 P\bigl[U \!\in\! \L^{x,D} \bigr]} \nonumber\\
			&\le\sqrt{10}\sqrt{ \max\left(\ell/|{\cal Y}|,q \ell/|{\cal Y}|, p_{triv}^{\mathfrak S}\right)}\le\sqrt{10}\sqrt{ \max\left(q \ell\cdot 2^{-n}, p_{triv}^{\mathfrak S}\right)},\nonumber
		\end{align}
	}{
\begin{align}
	\QTC{ \SZ \backslash \SUC \backslash \CL}{\SUC}&\leq \max_{x,D}\sqrt{10 P\bigl[U \!\in\! \L^{x,D} \bigr]} \nonumber\\
	&\le\sqrt{10}\sqrt{ \max\left(\frac{\ell}{|{\cal Y}|},\frac{q \ell}{|{\cal Y}|}, p_{triv}^{\mathfrak S}\right)}\nonumber\\
	&\le\sqrt{10}\sqrt{ \max\left(q \ell\cdot 2^{-n}, p_{triv}^{\mathfrak S}\right)},\nonumber
\end{align}
}
where we have used Equations \eqref{eq:case1}, \eqref{eq:case2} and \eqref{eq:case3} in the second inequality. Combining with Equations \eqref{eq:collision} and \eqref{eq:splitthecapacity} yields the desired bound.
		\qed
	\end{proof}

	\subsection{Online Extractability of the Fiat-Shamir \Trafo}\label{sec:OnExFS}
	
	We are now ready to state and proof the claimed online-extractability result for the \FST of (ordinary) \CnO protocols. 
	
	\begin{theorem}\label{thm:extractor}
		Let $\Pi$ be a $\frak{S}$-sound$^*$ ordinary \CnO protocol  with challenge space ${\cal C}_\lambda$ and $\ell = \ell(\lambda)$ commitments, and set $\kappa = \kappa(\lambda)  := \max_{c\in{\cal C}_\lambda}|c|$. Then, ${\sf FS}[\Pi]$ is a \PoKOE in the QROM (as in Definition~\ref{def:PoKOnline}), with $ \varepsilon_\text{\rm sim}(\lambda,q,n) = 0$ and 
		\begin{align*}
			\varepsilon_\text{\rm ex}(\lambda,q,n)  &\leq 2(\kappa+1) \cdot 2^{-n} + \bigg(2eq^{3/2}2^{-n/2} + q \sqrt{10 \max\left(q \ell\cdot 2^{-n}, p_{triv}^{\mathfrak S}\right)}\bigg)^2 \\
			&\leq (22\ell+60
			)q^{3}2^{-n} + 20 q^2 p_{triv}^{\mathfrak S} \, .
		\end{align*}
		The runtime of the extractor is dominated by running the compressed oracle, which has complexity $O(q^2) \cdot poly(n,B)$, and running~${\cal E}^*$.
	\end{theorem}
	We note that the above bound on $\varepsilon_\text{\rm ex}$ is asymptotically tight, except for the factor $\ell$. Indeed, the binding property of the hash-based commitment can be invalidated by means of a collision finding attack, which succeeds with probability $\Omega(q^{3}/2^{n})$. Furthermore the trivial soundness attack, which potentially applies to a $\frak{S}$-sound$^*$ \CnO protocol $\Pi$, can be complemented with a Grover search, yielding an attack against ${\sf FS}[\Pi]$ that succeeds with probability $\Omega(q^2 p_{triv}^{\mathfrak S})$. The non-tightness by a factor of $\ell$ is very mild in most cases. In particular, the number of commitments $\ell$ is polynomial in $\lambda$ and thus in $n$. For the most common case of a parallel repetition of a protocol with a constant number of commitments, using a hash function with output length linear in $\lambda$ (e.g. $n=3\lambda$) results in $\ell=O(n)=O(\lambda)$. 
	

	\begin{proof}
		We consider an arbitrary but fixed $\lambda \in \N$.
		For simplicity, we assume that $|c|$ is the same for all $c \in {\cal C}_{\lambda}$, and thus equal to $\kappa = \kappa(\lambda)$. If it is not, we could always make the prover output a couple of dummy outputs $m_i$ to match the upper bound on $|c|$. Let $\cal P^*$ be a dishonest prover that, after making $q$ queries to a RO $H$, outputs $(\instance,\pi)=(\instance,{\bf y},{\bf m}_\circ)$ 
		plus some (possibly quantum) auxiliary output $Z$. In the experiment ${\cal V}^{\cal E} \circ {\cal P}^*{}^{\cal E}(\lambda)$, our extractor $\cal E$ works as follows while simulating all queries to $H$ (by $\cal P^*$ and $\cal V$) with the compressed oracle: 
		\begin{enumerate}
			\item Run $\cal P^*(\lambda)$ to obtain $(\instance,\pi,Z)$ where $\pi=({\bf y},{\bf m}_\circ)$ with ${\bf m}_\circ = (m_1,\ldots,m_\kappa)$.
			\item Run ${\cal V}(\lambda,\instance,\pi)$ to obtain $v$. In detail: obtain $h_0 := H(\instance,{\bf y})$ and $h_j := H(m_j)$ for $j \in \{1,\ldots,\kappa\}$, and  set $v := \tt accept$ if and only if the pair consisting of ${\bf x} = \big( (\instance, {\bf y}),m_1,\ldots,m_\kappa \big)$ and ${\bf h} = (h_0,h_1,\ldots,h_\kappa)$ satisfies the relation $\tilde R$, defined to hold if and only if
			$$
			(h_1,\ldots,h_\kappa) = {\bf y}_c  \quad \wedge \quad V(\instance,c,{\bf m}_\circ)  \quad \text{where} \quad c := \gamma(h_0) \, .
			$$
			\item Measure the internal state of the compressed oracle to obtain $D$.
			\item Run ${\cal E}^*(\instance,{\bf m})$ on input $\instance$ and ${\bf m} := D^{-1}(\bf y)$ to obtain $w$. 
		\end{enumerate}
		Note that in the views of both  $\cal P^*$ and $\cal V$, the interaction with $H$ and the interaction with $\cal E$ differ only in that their oracle queries are answered by a compressed oracle instead of a real random-oracle in the latter case. This simulation is perfect and therefore $\varepsilon_\text{\rm sim}(\lambda,q,n) = 0$.
		
		Considering ${\cal P}^*$ as the algorithm $\cal A$ in Lemma~\ref{lemma:link}, the additional classical oracle queries that $\cal V$ performs in ${\cal V} \circ {\cal P}^*$ then match up with the algorithm $\tilde{\cal A}$, with $h_0,\ldots,h_\kappa$ here playing the role of $y_1,\ldots,y_\ell$ in Lemma~\ref{lemma:link}. 
		Thus,  
		\begin{align*}
			\Pr\bigl[  {\bf h} \neq D({\bf x}) \bigr]  
			\leq 2(\kappa(\lambda)+1)\cdot 2^{-n} \, .
		\end{align*}
		Therefore, we can bound the figure of merit $\varepsilon_\text{\rm ex}$ as
	\shortorfullversion{
\begin{align*}
			\varepsilon_\text{\rm ex}(\lambda,q,n) &= \Pr\bigl[v = {\tt accept} \,\wedge\, (\instance,w)\notin R\bigr] \\
			&= \Pr\bigl[  ({\bf x}, {\bf h}) \in \tilde R \,\wedge\, (\instance,w)\notin R\bigr] \\
			&\leq \Pr\bigl[ \big({\bf x}, D({\bf x})\big) \in \tilde R \,\wedge\, (\instance,w)\notin R\bigr] + 2(\kappa(\lambda)+1)\cdot 2^{-n} \\
			&\leq \Pr[\big({\bf x}, D({\bf x})\big) \in \tilde R \,\wedge\, (\instance,w)\notin R \,|\, D \not\in \SUC \cup \CL] \\
			&\qquad+ \Pr[D \in \SUC \cup \CL] + 2(\kappa(\lambda)+1)\cdot 2^{-n} \, .
		\end{align*}
}{
	\begin{align*}
		&\varepsilon_\text{\rm ex}(\lambda,q,n) = \Pr\bigl[v = {\tt accept} \,\wedge\, (\instance,w)\notin R\bigr] = \Pr\bigl[  ({\bf x}, {\bf h}) \in \tilde R \,\wedge\, (\instance,w)\notin R\bigr]\\
		& \leq \Pr\bigl[ \big({\bf x}, D({\bf x})\big) \in \tilde R \,\wedge\, (\instance,w)\notin R\bigr] + 2(\kappa(\lambda)+1)\cdot 2^{-n} \\
		&\leq \!\!\Pr[\!\big(\!{\bf x},\! D({\bf x})\big) \!\in\! \tilde R \! \wedge\!(\instance,\!w)\!\notin\! R |D \!\not\in \!\SUC \!\cup\! \CL]\!+ \!\Pr[D \!\in\! \SUC \!\cup\! \CL]\! +\! 2(\kappa(\lambda)\!+\!1)\!\cdot\! 2^{-n} \!\!\!\!\!\!.
	\end{align*}
	}
		Using the definition of $\tilde R$, understanding that $c := \gamma \circ D(\instance,{\bf y})$, we can write the first term as 
		\shortorfullversion{
		\begin{align*}
			\Pr\bigl[ &D({\bf m}_{\circ}) = {\bf y}_c \, \wedge \, V(\lambda, \instance, c, {\bf m}_{\circ}) \,\wedge\, (\instance,w)\notin R \,|\, D \not\in \SUC \cup \CL\bigr]\\
			&\leq \Pr\bigl[ V(\lambda, \instance, c, {\bf m}_c) \text{ for } {\bf m}:= D^{-1}({\bf y}) \,\wedge\, (\instance,w)\notin R \,|\, D \not\in \SUC \cup \CL\bigr]\\
			&\leq\Pr\bigl[D\in\SUC \,|\, D \not\in \SUC \cup \CL\bigr]  = 0 \, ,
		\end{align*}
	}{
	\begin{align*}
			\Pr\bigl[ &D({\bf m}_{\circ}) = {\bf y}_c \, \wedge \, V(\lambda, \instance, c, {\bf m}_{\circ}) \,\wedge\, (\instance,w)\notin R \,|\, D \not\in \SUC \cup \CL\bigr]\\
			&\leq \Pr\bigl[ V(\lambda, \instance, c, {\bf m}_c) \text{ for } {\bf m}:= D^{-1}({\bf y}) \,\wedge\, (\instance,w)\notin R \,|\, D \not\in \SUC \cup \CL\bigr]\\
			&\leq\Pr\bigl[D\in\SUC \,|\, D \not\in \SUC \cup \CL\bigr] \\
			& = 0 \, ,
		\end{align*}
	}
		where the first equality
		exploits that $D(m) = y$ iff $m = D^{-1}(y)$ for $D \not\in \CL$.

		We may thus conclude that
		\begin{align*} \varepsilon_\text{\rm ex}(\lambda,q,n)
			&\leq (2\kappa(\lambda)+1)\cdot 2^{-n} + \Pr\bigl[D\in \SUC\cup \CL\bigr] \\
			&\leq (2\kappa(\lambda)+1)\cdot 2^{-n} + \qQTC{\bot}{\SUC \cup \CL}^2 \, ,
		\end{align*}
		\shortorfullversion{using Eq. \eqref{eq:QTC} in the last inequality}{where the last inequality is by definition (\ref{eq:QTC}) of $\qQTC{\bot}{\cdot\,}$}. The \onlyfullversion{claimed} bound now follows from Lemma~\ref{lem:SUCCLbound}. 
		%
		\qed
	\end{proof}

	\subsection{The Unruh-\Trafo with a Compressing Hash Function} 
	\label{sec:UnruhWithCompression}
	
	We conclude this section by showing an improvement to the {\em Unruh \trafo}~\cite{Unruh2015}, which follows directly from our result above. 
	At the core of the Unruh \trafo is a generic technique to transform any \sigp into a \CnO protocol.  In~\cite{Unruh2015}, this \trafo is presented in combination with parallel repetition and the \FST as a means to construct (online-extractable) NIZK 
proofs of knowledge in the QROM. The entire \trafo was  later dubbed the Unruh \trafo.
	
\shortorfullversion{T}{	In fact, t}he Unruh \trafo was the first NIZK proof of knowledge in the QROM; the QROM security of the \FST was only established \shortorfullversion{much}{several years} later \cite{DFMS19,LZ19}. Despite being significantly less efficient than the \FST, the Unruh \trafo is still useful in certain cases because it puts weaker requirements on the underlying \sigp. 
	
	Here, to allow for a modular analysis, we consider the first step of the Unruh \trafo, i.e., the \trafo from a \sigp into a \CnO protocol, as an individual \trafo, which we refer to as the {\em pre-Unruh \trafo}, formally defined below. We stress that we allow the RO $H$ to be {\em compressing}, i.e. $|{\cal Y}| < |{\cal X}|$, while the extraction technique of \cite{Unruh2015} required $H$ to be a {\em length-preserving} RO. This obviously has a significant positive impact on the efficiency of the Unruh \trafo. 

	Let $\Sigma = ({\cal P}_\circ,{\cal V}_\circ)$ be a \sigp. We write $a_\circ \leftarrow {\cal P}_\circ$ to denote the first message in $\Pi_\circ$ as produced by ${\cal P}_\circ$ (for a given instance $\instance$). Furthermore, we write $z(a_\circ,c)$ for ${\cal P}_\circ$'s response then upon receiving challenge $c \in \cal C$.\footnote{We note that $z(a_\circ,c)$ may be a {\em randomized} function of $a_\circ$ and $c$. Furthermore, $z(a_\circ,c)$ is typically computed by ${\cal P}_\circ$ by means of the {\em randomness used to produce} $a_\circ$. } 
	
	\begin{definition}[Pre-Unruh \trafo]\label{def:pU}
		Let $\Sigma =({\cal P}_\circ,{\cal V}_\circ)$ be a \sigp as above. Then, the pre-Unruh-\trafo $\mathsf{pU}[\Sigma]=({\cal P},{\cal V})$ of $\Pi_\circ$ is the \CnO protocol with first message 
		$$
		a := (a_\circ, (y_i)_{i \in\mathcal C})
		$$
		where $a_\circ \leftarrow {\cal P}_\circ$ and for each $i \in \mathcal{C}$, $y_i := H(z_i)$ for $z_i := z(a_\circ,i))$, and with response $z: = z_c$ upon challenge $c \in \cal C$. 
		To verify, ${\cal V}$ runs ${\cal V}_\circ$ on $(a_\circ,c,z)$ and checks if $H(z)=y_c$; if both are true, it accepts, otherwise it rejects. 
	\end{definition}
	%
	Clearly, $\mathsf{pU}[\Sigma]$ is only efficient if $\Sigma$ has at most polynomially many possible challenges (which can always be obtained by restricting the challenge space). As mentioned, the resulting \CnO protocol can then be repeated in parallel and made non-interactive using the \FST. We will now provide a fairly straightforward corollary to conclude the security of the more efficient variant of the (full) Unruh \trafo that allows for a compressing RO, given by the composition of the pre-Unruh \trafo introduced above, parallel repetition and the \FST. In the following, denote the $r$-fold parallel repetition of a (\CnO) \sigp $\Pi$ by $\Pi^r$ and use the notation $\mathsf{Unr}_r[\Sigma] := \mathsf{FS}\left[\mathsf{pU}[\Sigma]^r\right]$ for the Unruh \trafo with $r$-fold parallel repetition. 
	\onlyfullversion{
		\begin{remark}
		A proof in $\Pi=\mathsf{Unr}_r[\Sigma]$ can be generated in time $T^\Pi_{\mathcal P}=rT^\Sigma_{\mathcal P}+(\ell_0 r+1) T_H$, and verified in time $T^\Pi_{\mathcal V}=rT^{\Sigma}_{\mathcal V}+(1+r)T_H$, where $T^{\Sigma}_{\mathcal P}, T^{\Sigma}_{\mathcal V}$ and $T_H$ are the prover and verifier runtime of $\Sigma$, and the time required for computing one hash, respectively. 
	\end{remark}
	}
	It is straightforward to verify that the pre-Unruh \trafo does not harm most security properties of the \sigp. In particular, it tightly preserves soundness 
	and honest-verifier zero-knowledge (in the QROM). It also preserves $\mathfrak S$-soundness in a certain sense.
	
	\begin{proposition}\label{prop:pU-is-nice-with-S-soundness}
		Let $\Sigma$ be an $\mathfrak S$-sound \sigp with challenge space size $\ell=\ell(\lambda)$ with extractor runtime $T$. Then $\Pi := \mathsf{pU}[\Sigma]$ is $\mathfrak S$-sound as a \CnO protocol with extractor runtime $T'\le T+O(\ell)$. Furthermore, suppose that membership in $\mathfrak S$ is checkable in time $T_{\mathfrak S}$. Then $\Pi$ is $\mathfrak S$-sound${}^*$ with extractor runtime $T''\le T'+\ell^2 T_{\mathfrak S}+\ell T_{\cal V}$, where $T_{\mathcal V}$ is the runtime of $\Pi$'s verification predicate $\mathcal V$. 
	\end{proposition}
	
	\begin{proof}
		Let $\mathcal E_{\Sigma}$ be the extractor for $\Sigma$ guaranteed to exist by Definition~\ref{def:S-sound-Ord}. Note that for $\Pi = \mathsf{pU}[\Sigma]$ regarded as a \CnO protocol, for each challenge exactly one of the commitments has to be opened. For such protocols, we use $c$ and $\{c\}$ interchangeably (where $c$ is a challenge in $\Pi$). We define an extractor $\mathcal E_{\Pi}$ as follows. On input $(\instance, m_1,...,m_\ell, a_{\circ}, S)$, run $w=\mathcal E_{\Sigma}(\instance,a_\circ, S, \{m_c\}_{c\in S})$, then output $w$. The only runtime overhead of $\mathcal E_{\Pi}$ results from having to parse its input and preparing the input for $\mathcal E_{\Sigma}$.
		
		We continue to define an $\mathfrak S$-soundness${}^*$ extractor $\mathcal E^*_{\Pi}$ for $\Pi$ as follows. On input $(\instance, m_1,...,m_\ell, a_{\circ})$, compute $b_c=\mathcal V(\instance, a_\circ, c, m_c)$ for all $c\in \mathcal C$, and set $\hat S=\{c\in\mathcal C \, |\, b_c=1\}$. Using at most $\ell(\ell+1)/2$ membership tests for $\mathfrak S$, find $S\subseteq \hat S$ such that $S\in\mathfrak S_{\min}$. Finally, run $w=\mathcal E_{\Pi}(\instance, m_1,...,m_\ell,a_{\circ}, S)$ and output $w$. The runtime statement is straightforward.
		\qed\end{proof}
	
	
	Using Proposition~\ref{prop:pU-is-nice-with-S-soundness} above and  Lemma~5.3 from \cite{DFMS21} to argue $\mathfrak S^{\vee r}$-soundness${}^*$ of the parallel repetition of $\mathsf{pU}[\Pi]$, and using Theorem \ref{thm:extractor} to argue online extractability of its \FST
	, we obtain the online-extractability of the Unruh \trafo with computationally binding commitments, i.e., when using a {\em compressing} hash function for the commitments.
	\begin{corollary}\label{cor:Unruh}
		Let $\Sigma$ be an $\mathfrak S$-sound \sigp 
		with challenge space size $\ell_{0}$. Then $\Pi := \mathsf{Unr}_r[\Sigma] = \mathsf{FS}[\mathsf{pU}[\Sigma]^r]$ is a \PoKOE in the QROM (as in Definition \ref{def:PoKOnline}) with  $\varepsilon_\text{\rm sim}=0$ and
		\begin{equation}
			\varepsilon_{\mathrm{ex}}(\lambda,q,n) \le (22r\ell_{0}+60)q^{3}2^{-n} + 20 q^2 \left(p_{triv}^{\mathfrak S}\right)^r \, .
		\end{equation}
		The online extractor \onlyfullversion{for $\Pi$ }runs in time $T^\Pi_{\mathcal E}\le rT^{\mathsf{pU}[\Sigma]}_{\mathcal E}+O(q^2) \cdot poly(n,B)$, where $T_{\mathcal E}^{\mathsf{pU}[\Sigma]}$ is the runtime of $\mathsf{pU}[\Sigma]$'s $\mathfrak S$-soundness${}^*$ extractor as given in Proposition~\ref{prop:pU-is-nice-with-S-soundness}. 
	\end{corollary}

	\section{Online Extractability of the FS-\Trafo: \\ The Case of \MCnO Protocols}\label{sec:MCnO}
	
	For an ordinary \CnO protocol with reasonable concrete security (e.g., 128 bits),
	the number of commitments $\ell$ might be considerable. 
	In this case, the communication complexity of  the protocol (and thus the size of the non-interactive proof system, or digital-signature scheme, obtained via the \FST) can be reduced by using a {\em Merkle tree} to collectively commit to the $\ell$ strings~$m_i$. Such a construction is mentioned in \cite{Fischlin05}, and it is used in the construction of the digital-signature schemes Picnic2 and Picnic3 \cite{KKW18,KZ20,CD+19}. The \MCnO mechanism shrinks the commitment information from $\ell\cdot n$ to $n$, at the expense of increasing the cost of opening $|c|$ values $m_i$ by an additive term of about $\lessapprox |c|\cdot n\cdot \log\ell$. 
	
	The cost of opening can, in fact, be slightly reduced again, by streamlining the opening information. When opening several leaves of a Merkle tree, the authentication paths overlap, so opening requires a number of hash values less than $h$ per leaf, where $h$ is the height of the tree. This overlap was observed and exploited in the octopus authentication algorithm which constitutes one of the optimizations of the stateless hash-based signature scheme gravity-SPHINCS \cite{AE18}, as well as in Picnic2 and Picnic3 \cite{KZ20}. In the following section, we formalize tree-based collective commitment schemes with ``octopus'' opening.

	\subsection{\MCnO Protocols}\label{sec:Octopus}

	In line with Remark~\ref{rem:generalCandO}, we can consider \CnO protocols with a different choice of commitment scheme, compared to the default choice of committing by element-wise hashing. Here, we discuss a particular choice of an alternative commitment scheme, which gives rise to more efficient \CnO protocols in certain cases when $\ell$ is large. Informally, we consider \CnO protocols where $m_1,\ldots,m_\ell$ is committed to by using a {\em Merkle tree}, and individual $m_i$'s are opened by announcing the corresponding authentication paths. 

To make this more formal, we introduce the following notation. 
		For simplicity, we assume that $\ell$ is a power of $2$, and thus $\ell = 2^h$ for $h \in \N$. We then consider the {\em full binary tree} $\Tree = \{0,1\}^{\leq h}$ of depth $h$, where the vertices are identified by bit strings. The root is denoted by $\emptyset$; the $i$-th leave is denoted by $\lf(i) \in \{0,1\}^h$ and is given by the binary representation of $i \in [\ell]$.  
	The {\em authentication path} for the $i$-th leaf is the subtree that consists of all the ancestors of $\lf(i)$ and their siblings: 
	\begin{align*}
		\Auth(i) := \mathsf{Anc}(\lf(i)) \cup \{ \mathsf{sib}(v) \,|\, \emptyset \neq v \in \mathsf{Anc}(\lf(i)) \} \, ,
	\end{align*}
	where $\mathsf{Anc}(v) := \{ u \in \Tree \,|\, \exists \,w \!:\! u\|w \!=\! v \}$ and $\mathsf{sib}(u\|b) := u\|(1-b)$ for any $b \in \{0,1\}$. 
	Finally, for any subset $c \subseteq \{1,\ldots,\ell\}$, we let $\Auth(c) := \bigcup_{i \in c} \Auth(i)$ be the union of the authentication paths of the considered leaves, and we define the {\em octopus} $\Octo(c)$ to be the restriction of $\Auth(c)$ to its leaves, but excluding the leaves $\lf(i)$ for $i \in c$, i.e., 
	$$
	\Octo(c) := \mathsf{leaves}(\Auth(c))  \setminus \{ \lf(i) \,|\, i \in c \} 
	$$ 
	where, for any subtree $T$ of $\Tree$, $\mathsf{leaves}(T) := \{ v \in T \,|\, (v\|0),(v\|1) \not\in T \}$. 
	
	Extending on the above notation, for a given hash function $H: {\cal X} \to {\cal Y}$, where ${\cal X} = \{0,1\}^{\leq B}$ and ${\cal Y} = \{0,1\}^n$ for sufficiently large $B$, we define the {\em Merkle tree} of ${\bf m} = (m_1,...,m_\ell) \in {\cal X}^\ell$ to be the {\em labeled} binary tree that has its leaves $\lf(1),\ldots,\lf(\ell)$ labeled by $H(m_1),...,H(m_\ell)$, respectively, and each internal vertex is labeled by the hash of the labels of its two children. 
	Formally, 
	$$
	\MTree_H({\bf m}) := \big\{\big(v,l_v({\bf m})\big) \,\big|\, v \in \Tree \big\}
	$$ 
	with the labeling $l_v({\bf m})$ recursively defined as 
	$$
	l_v({\bf m}) := H\big(l_{v\|0}({\bf m}) \| l_{v\|1}({\bf m})\big) \:\text{ for }\: v \in \{0,1\}^{< h} \quad$$ and
	$$l_{\lf(i)}({\bf m}) := H(m_i) \:\text{ for }\: i \in \{1,\ldots,\ell\} \, , 
	$$
	where we leave the dependency of the labeling on $H$, i.e., $l_v = l_v^H$, implicit. 
	We also write $\MRoot_H({\bf m})$ then for the root label $l_{\emptyset}({\bf m})$. 
	In the same spirit, we write $\MAuth_H(c, {\bf m}) := \big\{\big(v,l_v({\bf m})\big) \,\big|\, v \in \Auth(c) \big\}$ for the labeled authentication path and $\MOcto_H(c, {\bf m}) := \big\{\big(v,l_v({\bf m})\big) \,\big|\, v \in \Octo(c) \big\}$ for the labeled octopus, using the same labeling function as for the Merkle tree. 
	
	\onlyshortversion{To make this more formal, we introduce the following notation (see Appendix~\ref{app:MerkleTree} for a formal discussion, and see Fig.~\ref{fig:Octopus} for an example) For simplicity, we assume that $\ell$ is a power of $2$. We write $\MTree_H({\bf m})$ for the Merkle tree of messages ${\bf m} = (m_1,\ldots,m_\ell)$ computed using hash function $H$; more formally, the (labels of the) vertices in the Merkle tree are recursively computed as $l_v({\bf m}) := H\big(l_{v\|0}({\bf m}) \| l_{v\|1}({\bf m})\big)$, with the leaves being the hashes of the $m_i$'s. $\MRoot_H({\bf m})$ then denotes the root of the Merkle tree. Furthermore, for $c \subseteq [\ell]$, we write $\MAuth_H(c, {\bf m})$ for the union of the authentication paths for all messages $m_i$ with $i \in c$, and the {\em octopus} $\MOcto_H(c, {\bf m})$ denotes all the vertices needed to compute all the authentication paths in $\MAuth_H(c, {\bf m})$, but excluding the hashes of the actual messages $m_i$ with $i \in c$ (see Fig.~\ref{fig:Octopus}). }

	\begin{figure}\centering
		\scalebox{0.4}{%
			
			\begin{tikzpicture}[label distance=4mm,scale=1.5,octo/.style = {circle,fill=black!=20,inner sep=.15cm}]
				\begin{pgfonlayer}{nodelayer}
					\node[octo]  (0) at (-6, 0) {};
					\node[octo,preaction={
						draw,yellow,inner sep = 0.2cm,
						double=yellow,
						double distance=0.3cm,
					}]  (1) at (-2, 0) {};
					\node[octo]  (2) at (2, 0) {};
					\node[octo]  (3) at (6, 0) {};
					\node[octo]  (4) at (-4, 2) {};
					\node[octo,preaction={
						draw,yellow,inner sep = 0.2cm,
						double=yellow,
						double distance=0.3cm,
					}]  (5) at (4, 2) {};
					\node[octo, label=above:\Huge$y$]  (6) at (0, 4) {};
					\node[octo,preaction={
						draw,yellow,inner sep = 0.2cm,
						double=yellow,
						double distance=0.3cm,
					},label=below:\Huge$H(m_2)$]  (7) at (-5, -2) {};
					\node[octo,label=below:\Huge$H(m_3)$]  (8) at (-3, -2) {};
					\node[octo,label=below:\Huge$H(m_1)$]  (9) at (-7, -2) {};
					\node[octo,label=below:\Huge$H(m_4)$]  (10) at (-1, -2) {};
					\node[octo,label=below:\Huge$H(m_5)$]  (11) at (1, -2) {};
					\node[octo,label=below:\Huge$H(m_6)$]  (12) at (3, -2) {};
					\node[octo,label=below:\Huge$H(m_7)$]  (13) at (5, -2) {};
					\node[octo,label=below:\Huge$H(m_8)$]  (14) at (7, -2) {};
				\end{pgfonlayer}
				\begin{pgfonlayer}{edgelayer}
					\draw (6.center) to (4.center);
					\draw (4.center) to (0.center);
					\draw (4.center) to (1.center);
					\draw (6.center) to (5.center);
					\draw (5.center) to (2.center);
					\draw (5.center) to (3.center);
					\draw (0.center) to (9.center);
					\draw (7.center) to (0.center);
					\draw (1.center) to (8.center);
					\draw (1.center) to (10.center);
					\draw (2.center) to (11.center);
					\draw (2.center) to (12.center);
					\draw (3.center) to (13.center);
					\draw (3.center) to (14.center);
				\end{pgfonlayer}
		\end{tikzpicture}}
		\caption{The Merkle tree $\MTree_H({\bf m})$ for ${\bf m} = (m_1,\ldots,m_8)$ with $\MRoot_H({\bf m}) = y$. The yellow vertices mark the octopus $\MOcto_H(\{1\},{\bf m})$, which is revealed (along with~$m_1$) when opening the commitment $y$ to $m_1$. 
		}\label{fig:Octopus}
	\end{figure}
	
	
	A {\em \MCnO}  protocol is now defined to be a variation of a \CnO protocol, as hinted at in Remark~\ref{rem:generalCandO}, where the first message of the protocol, i.e., the commitment of ${\bf m} = (m_1,\ldots,m_\ell)$, is computed as $y = \MRoot_H({\bf m})$, and the response $z$ for challenge-set $c$ then consists of the messages ${\bf m}_c = (m_i)_{i \in c}$ together with $O = \MOcto_H(c,{\bf m})$. The verifier $\cal V$ then accepts if and only if ${\bf m}_c$ and $O$ ``hash down to'' $y$ and the predicate $V(\lambda,\instance,c,{\bf m}_c,a)$ is satisfied. More formally, the former means that $\cal V$ computes $\MAuth_H(c, {\bf m})$ from $O \cup \{(\lf(i),H(m_i)) \,|\, i \in c \}$ in the obvious way, and then checks whether $l_{\emptyset}({\bf m}) = y$.  This verification is denoted by $\OctoVerify^H(c,y,{\bf m}_c,O)$, see Fig.~\ref{fig:merklesigp}\onlyshortversion{ in Section~\ref{sec:sigpfigures}}.
	
		\begin{figure}\centering
	\begin{tikzpicture}
	\begin{pgfonlayer}{nodelayer}
	\node  (0) at (-19, 3) {};
	\node  (1) at (-15, 3) {};
	\node  (2) at (-15, 2) {};
	\node  (3) at (-19, 2) {};
	\node  (4) at (-19, 1) {};
	\node  (5) at (-15, 1) {};
	\node  (7) at (-14, 3.75) {$\cal V$};
	\node  (8) at (-20, 3.75) {$\cal P$};
	\node  (9) at (-17, 3.25) {$a_\circ, y = \MRoot_H({\bf m})$};
	\node  (10) at (-17, 2.25) {$c$};
	\node  (11) at (-17, 1.25) {${\bf m}_c, O=\MOcto_H(c, {\bf m})$};
	\node  (35) at (-13.25, 0.75) {};
	\node  (in) at (-11.5, 0.875) {$\OctoVerify^H(c,y,{\bf m}_c,O) \,\wedge\, V(\instance,c,{\bf m}_c,a_\circ)$};
	\node  (40) at (-13.25, 2.25) {$c \leftarrow{\cal C} \subseteq 2^{[\ell]}$};
	\end{pgfonlayer}
	\begin{pgfonlayer}{edgelayer}
	\draw [style=arrow] (0.center) to (1.center);
	\draw [style=arrow] (2.center) to (3.center);
	\draw [style=arrow] (4.center) to (5.center);
	\end{pgfonlayer}
	\end{tikzpicture}
	\caption{A Merkle-tree based C\&O \sigp, formally introduced in Section~\ref{sec:Octopus}.\label{fig:merklesigp}}
\end{figure}
	
	Looking ahead, we may also consider a variation where the verifier resamples the challenge $c$ if the resulting octopus is bigger than a given bound. Formally, this means that the challenge space of the \MCnO  protocol is restricted to those challenges $c \in [\ell]$ for which $\Octo(c)$ is not too large.

	\subsection{Online Extractability of the Fiat-Shamir \Trafo}
	
	The analysis in Section \ref{sec:FSCnO} can be generalized to the case of FS-transformed \MCnO protocols. To that end, we generalize the notation from that section as follows. Let $\Pi$ be a \MCnO protocol with number of messages to be committed equal to $\ell=2^h$ where $h$ is the height of the commitment Merkle tree.\footnote{As in the previous section we assume that $\ell$ is a power of 2 for ease of exposition.} 
	
	For a given database $D \in \DB$, recall from Section~\ref{sec:FSCnO} the definition of $D^{-1}$; 
	applied to a tuple ${\bf y} = (y_1,\ldots,y_\ell) \in {\cal Y}^\ell$ of commitments, $D^{-1}$ attempts to recover the corresponding committed messages $m_1,\ldots,m_\ell$. Here, in a similar spirit but now considering the Merkle-tree commitment, $\MRoot_D^{-1}$ attempts to recover the committed messages from the root label 
	of the Merkle tree. 

	In more detail, for a commitment $y \in {\cal Y} = \{0,1\}^{n}$ we reverse engineer the Merkle tree in the obvious way (see Fig.~\ref{partialMTree}\onlyshortversion{ in Appendix~\ref{App:OmittedProofs}} for an example); namely, accepting a small clash in notation with the labeling function $l_v({\bf m})$ defined for a tuple ${\bf m} \in {\cal M}^\ell$, we set the root label $l_\emptyset(y) := y$, and recursively define 	
	$$
	\big(l_{v\|0}(y) , l_{v\|1}(y)\big)  := \mathsf{split} \circ D^{-1}\big(l_v(y)\big) \in {\cal Y} \times {\cal Y} 
	$$
	for $\emptyset \neq v \in \{0,1\}^{\leq h}$, where $\mathsf{split}$ maps any $2n$-bit string, parsed as $y_1\|y_2$ with $y_1,y_2 \in \{0,1\}^{n}$, to the pair $(y_1,y_2)$ of $n$-bit strings, while it maps anything else to $(\bot,\bot)$.
	Then, accepting a small clash in notation again, we set 
	$$
	\MTree_D(y) := \{ l_v(y) \,|\, v \in \{0,1\}^{\leq h} \} \, ,
	$$
	and finally\onlyshortversion{, with $\lf(i)$ denoting the $i$-th leaf in the tree, }
	$$
	\MRoot_D^{-1}(y) := \big(D^{-1}\bigl(l_{\lf(1)}(y)\bigr),\ldots,D^{-1}\bigl(l_{\lf(\ell)}(y)\bigr)\big) \, .
	$$
	%
	Following the strategy we used in Section~\ref{sec:FSCnO}, we define the database property 
	\begin{equation*} 
		\SUC := \left\{D \,\bigg|\, \begin{array}{c}\exists\, y \in {\cal Y}\text{ and }\instance\in \INST\text{ so that } {\bf m}:= \MRoot_D^{-1}(y)\text{ satisfies} \\
			V(\instance,c,{\bf m}_c) \text{ for } c := \gamma \circ D(\instance, y)  \;\text{and}\;
			\big(inst,{\cal E}^*(\instance,  {\bf m})\big) \not\in R
		\end{array} \right\} ,
	\end{equation*}
	and our first goal is to show that $\qQTC{\bot}{\SUC \cup \CL}$ is small. 
	
	\shortorfullversion{
	\begin{lemma}\label{lem:SUC-Oc-CLbound}$
		\qQTC{\bot}{\SUC \cup \CL} \leq 2eq^{3/2}2^{-n/2} + q\sqrt{10 \max\left(q \ell\cdot 2^{-n+1}, p_{triv}^{\mathfrak S}\right)} \, .
		$
	\end{lemma}
}{\begin{lemma}\label{lem:SUC-Oc-CLbound}
		Let $\Pi$ be an $\mathfrak{S}$-sound \CnO protocol with $p^{\mathfrak S}_{triv}$ as defined in~\eqref{eq:ptriv}. Then
		$$
		\qQTC{\bot}{\SUC \cup \CL} \leq 2eq^{3/2}2^{-n/2} + q\sqrt{10 \max\left(q \ell\cdot 2^{-n+1}, p_{triv}^{\mathfrak S}\right)} \, .
		$$
	\end{lemma}}
	The proof works exactly as the proof of Lemma \ref{lem:SUCCLbound}, accounting for some syntactic differences due to the Merkle tree commitment. In particular, where in Case 1 and 2 of the proof of Lemma \ref{lem:SUCCLbound} we have to exclude $U$ from falling on one of the hash values $y_1,\ldots, y_\ell$ in order to keep the $\bf m$ that was constructed from the database intact, we now have a similar restriction for $U$, but with respect to the whole tree $\MTree_D(y)$. 
\onlyshortversion{		The full proof can be found in Appendix \ref{App:OmittedProofs}.}

	\begin{proof}
		As in the proof of of Lemma \ref{lem:SUCCLbound}, we can bound
		\begin{align}
		\qQTC{\bot}{\SUC \cup \CL} &\leq  \sum_{s=0}^{q-1} \big(\QTC{\SZ \backslash \CL}{\CL} +  \QTC{ \SZ \backslash \SUC}{\SUC} \big) \label{eq:Merksplitthecapacity}
		\end{align}
		and use that
		\begin{equation}\label{eq:Merkcollision}
		\QTC{\SZ \backslash \CL}{\CL} \leq 2e\sqrt{(s+1)/2^n} \leq 2e\sqrt{q/2^n} \, .
		\end{equation}
		Thus, it remains to control the second term, which we will do again by means of Theorem~\ref{thm:simple} with $\P := \SZ\!\setminus \SUC$ and $\P' := \SUC$. 
		
		To this end, we consider arbitrary but fixed $D \in \DB$ and input $x \in \cal X$. By Remark~\ref{rem:IndOfD(x)}, we may assume that $D(x) = \bot$. Furthermore, for $\P|_{D|^x}$ to be non-empty, it must be that $D \in \SZ$, i.e., $D$ is bounded in size. We now distinguish between the following cases for the considered $D$ and $x$.
		
		\paragraph{Case 1:} $D \in \SUC$. In particular, $\bot \in \SUC|_{D|^x} = \P'_{D|^x}$. So, Theorem~\ref{thm:simple} instructs us to set $\L:= \P_{D|^x}$, where we leave the dependency of $\L$ on $D$ and $x$ implicit. Given that $D \in \SUC$, we can consider $\instance$ and $y$ as promised by the definition of $\SUC$ above, i.e., such that
		$V(\instance,c, {\bf m}_{c})$ and $(\instance,{\cal E}^*(\instance, {\bf m}))  \notin R$
		for
		\begin{equation}\label{eq:EqualCandM}
		c := \gamma \circ D(\instance, y)  \quad\text{and}\quad {\bf m}:= \MRoot_D^{-1}(y) \, .
		\end{equation}
		Note that, since $D(x) = \bot$ and $V(\instance,c, {\bf m}_{c})$ holds, which in particular means that $c$ must be defined, it must be that $x \neq (\instance,y)$. Therefore
		\begin{equation}\label{eq:EqualC}
		\gamma \circ D(\instance,y) = \gamma\circ D[x \mapsto u](\instance,y) \, .
		\end{equation}
		Our goal now is to show the final implication in 
		$$
		u \in \L \:\Longleftrightarrow\: D[x \!\mapsto\! u] \in \P \:\Longrightarrow\: D[x \!\mapsto\! u] \not\in \SUC  \:\Longrightarrow\: u\in \MTree_D(y) \, .
		$$
		We will do this by showing that $u\notin \MTree_D(y)$ implies
		\begin{equation}\label{eq:EqualM}
		\MRoot_D^{-1}(y) = \MRoot_{D[x \mapsto u]}^{-1}(y) \, .
		\end{equation}
		Indeed, the contraposition $u\notin \MTree_D(y) \Rightarrow D[x \!\mapsto\! u] \in \SUC$ of the claimed implication then follows from the fact that (\ref{eq:EqualC}) and (\ref{eq:EqualM}) together imply that $c$ and $\bf m$ remain unchanged when replacing $D$ by $D[x \!\mapsto\! u]$ in (\ref{eq:EqualCandM}), and so $D[x \!\mapsto\! u] \in \SUC$ as well. 

		Towards showing (\ref{eq:EqualM}), 
		exploiting again that $D(x)=\bot$, it follows by definition of the reverse engineered labeling function $l_v(y)$ that $x\neq (l_{v||0}(y),l_{v||1}(y))$ for any $v$ with $l_{v||0}(y) \neq \bot \neq l_{v||1}(y)$, i.e., $x$ is not equal to any pair of siblings in $\MTree_D(y)$ with non-$\bot$ labeling 
		(see Figure~\ref{partialMTree}). 
				Due to a similar reasoning, $x \neq m_i$ for any $i$. 
		It now follows by definition of the reverse engineered Merkle tree and of $\MRoot^{-1}$ that if $u \notin \MTree_D(y)$ then $\MTree_D(y) = \MTree_{D[x\mapsto u]}(y)$ and $\MRoot_D^{-1}(y) = \MRoot_{D[x \mapsto u]}^{-1}(y)$, as claimed.

		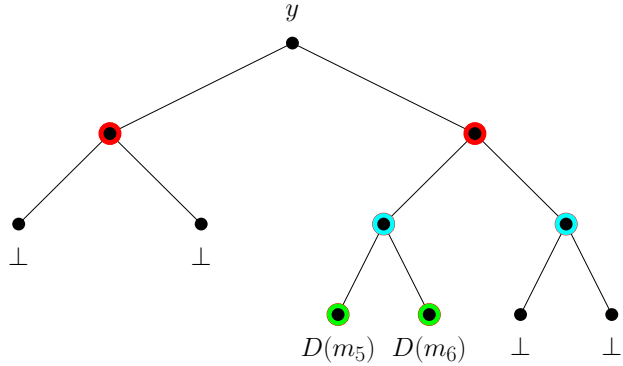
\begin{figure}\centering
			\scalebox{0.4}{%
				\pgfdeclarelayer{nodelayer}
				\pgfdeclarelayer{edgelayer}
				\pgfsetlayers{edgelayer,nodelayer}
				\begin{tikzpicture}[label distance=4mm,scale=1.5,octo/.style = {circle,fill=black!=20,inner sep=.15cm}]
				\begin{pgfonlayer}{nodelayer}
				\node[octo,label=below:\Huge$\bot$]  (0) at (-6, 0) {};
				\node[octo,label=below:\Huge$\bot$]  (1) at (-2, 0) {};
				\node[octo,preaction={
					draw,red,inner sep = 0.2cm,
					double=aqua,
					double distance=0.3cm,
				}]  (2) at (2, 0) {};
				\node[octo,preaction={
					draw,red,inner sep = 0.2cm,
					double=aqua,
					double distance=0.3cm,
				}]  (3) at (6, 0) {};
				\node[octo,preaction={
					draw,red,inner sep = 0.2cm,
					double=red,
					double distance=0.3cm,
				}]  (4) at (-4, 2) {};
				\node[octo,preaction={
					draw,red,inner sep = 0.2cm,
					double=red,
					double distance=0.3cm,
				}]  (5) at (4, 2) {};red
				\node[octo,label=above:\Huge $y$]  (6) at (0, 4) {};
				\node[octo,preaction={
					draw,red,inner sep = 0.2cm,
					double=green,
					double distance=0.3cm,
				},label=below:\Huge$D(m_5)$]  (11) at (1, -2) {};
				\node[octo,preaction={
					draw,red,inner sep = 0.2cm,
					double=green,
					double distance=0.3cm,
				},label=below:\Huge$D(m_6)$]  (12) at (3, -2) {};
				\node[octo,label=below:\Huge$\bot$]  (13) at (5, -2) {};
				\node[octo,label=below:\Huge$\bot$]  (14) at (7, -2) {};
				\end{pgfonlayer}
				\begin{pgfonlayer}{edgelayer}
				\draw (6.center) to (4.center);
				\draw (4.center) to (0.center);
				\draw (4.center) to (1.center);
				\draw (6.center) to (5.center);
				\draw (5.center) to (2.center);
				\draw (5.center) to (3.center);
				\draw (2.center) to (11.center);
				\draw (2.center) to (12.center);
				\draw (3.center) to (13.center);
				\draw (3.center) to (14.center);
				\end{pgfonlayer}
				\end{tikzpicture}}
			\caption{Example of a reverse engineered Merkle tree $\MTree_D(y)$, with the $\bot$-children of the $\bot$-labels omitted. Since $D(x) = \bot$, $x\neq (l_{u}(y),l_{w}(y))$ for any two siblings $(u,w)$ in $\MTree_D(y)$, i.e., nodes with the same color. Assuming that $u \not\in \MTree_D(y)$ then implies that reprogramming $D$ to $D[x\mapsto u]$ does not affect the reverse engineered Merkle tree.  }\label{partialMTree}
		\end{figure}

		Thus, we can bound
		\begin{equation}\label{eq:Merkcase1}
		P[U \!\in\! \L] \leq P[ U \!\in\! \MTree_D(y)] \leq \frac{2 \cdot 2^h-1}{|{\cal Y}|}=\frac{2\ell-1}{|{\cal Y}|} \, .
		\end{equation}
		
		\paragraph{Case 2:} $D \not\in \SUC$, and $x$ is a ``commit query'', i.e., $x=m \in \cal{M}$ or $x = (l_{v\|0},l_{v\|1})$ for two labels $l_{v\|0},l_{v\|1} \in \cal Y$. In particular, $\bot \not\in \P'|_{D|^x}$ (given that $D(x)=\bot$) and so in the light of Theorem~\ref{thm:simple} we may choose $\L := \P'|_{D|^x}$. We then have
		\begin{equation*}
		u \in \L \:\Longleftrightarrow\: D[x \!\mapsto\! u] \in \P' = \SUC  \:\Longrightarrow\:  \exists \, \instance,y:  D(\instance, y) \neq \bot \wedge u\in \MTree_D(y) 
		\end{equation*}
		where final implication can be seen as follows. 
		By definition of $\SUC$, the assumption $D[x \!\mapsto\! u] \in \SUC$ implies the existence of $\instance$ and $y$ with
		$V(\instance,c, {\bf m}_{c})$ and $\big(\instance,{\cal E}^*(\instance,  {\bf m})\big) \not\in R$
		for
		$$
		c := \gamma \circ D[x \!\mapsto\! u](\instance, y) = \gamma \circ D(\instance, y)  \quad\text{and}\quad {\bf m}: = \MRoot_{D[x \mapsto u]}^{-1}(y) \, ,
		$$
		where the equality in the definition of $c$ exploits that $x$ is not a ``challenge'' query. The fact that $V(\instance,c,{\bf m_c})$ is satisfied for this $c$ thus implies that $D(\instance,y)\neq \bot$.
		Next, with the goal to reach a contradiction, assume that $u \notin \MTree_D(y)$. Then for all $\bot \neq h\in \MTree_D(y)$ we have that $D^{-1}(h) = D[x\!\mapsto\!u]^{-1}(h)$ except if $D(x) = h$, but this cannot be since $D(x) = \bot$. It follows that $\MTree_D(y) = \MTree_D[x \!\mapsto\! u](y)$ and $\MRoot_{D}^{-1}(y) = \MRoot_{D[x \mapsto u]}^{-1}(y)$.
		The above then implies that $D \in \SUC$, a contradiction. 
		
		Thus, we can bound
		\begin{equation}\label{eq:Merkcase2}
		P[U \!\in\! \L] \leq P[\, \exists \, \instance,y :  D(\instance,y) \neq \bot\wedge U\in \MTree_D(y)] \onlyfullversion{\leq \frac{s(2\ell - 1)}{|{\cal Y}|}} \leq \frac{q(2\ell - 1)}{|{\cal Y}|} \, .
		\end{equation}
		
		\paragraph{Case 3:} $D \not\in \SUC$, and $x$ is a ``challenge query'', i.e., $x = (\instance, y) \in \INST \times {\cal Y}$. Set ${\bf m}: = \MRoot_{D}^{-1}(y) $. Again, we have that $\bot \not\in \SUC|_{D|^x} = \P'_{D|^x}$, and so by Theorem~\ref{thm:simple} we may set $\L:= \P'_{D|^x}$.
		Here, we can argue that 
		\begin{align*}
		u \in \L \:\Longleftrightarrow\: &D[x \!\mapsto\! u] \in \P' =\SUC \\
		\Longrightarrow\:  &V(\instance,\gamma(u), {\bf m} _{{\gamma(u)}})\text{ and }\big(\instance,{\cal E}^*(\instance,  {\bf m})\big) \not\in R \, ,
		\end{align*} 
		where the final implication can be seen as follows. 
		By definition of $\SUC$, the assumption $D[x \!\mapsto\! u] \in \SUC$ implies the existence of $\instance'$ and $y'$ with
		$V(\instance',c, {\bf m}'_{c})$ and $\big(\instance',{\cal E}^*(\instance',  {\bf m}')\big) \not\in R$
		for 
		$$
		c := \gamma \circ D[x \!\mapsto\! u](\instance', y')  \quad\text{and}\quad {\bf m}': = \MRoot_{D[x\mapsto u]}^{-1}(y') = \MRoot_D^{-1}(y') \, ,
		$$
		where the very last equality exploits that $x$ is not a ``commit'' query. 
		With the goal to come to a contradiction, assume that $(\instance',y') \neq (\instance, y) = x$. Then, $c  = \gamma \circ D[x \!\mapsto\! u](\instance', y') = \gamma \circ D(\instance', y')$, and the above then implies that $D \in \SUC$, a contradiction. Thus, $(\instance',y') = (\instance, y) = x$. In particular, ${\bf m}' = {\bf m}$ and $c = \gamma \circ D[x \!\mapsto\! u](\instance', y') = \gamma \circ D[x \!\mapsto\! u](x) = \gamma(u)$. Hence, the claimed implication holds. 
		
		Thus, we can bound 
		\begin{align}
		P[U \!\in\! \L] &\leq P[V(\instance,\gamma(U), {\bf m}_{\gamma(U)}) \, \wedge \, \big(\instance,{\cal E}^*(\instance,  {\bf m})\big) \not\in R] \nonumber\\[0.5ex]
		&\leq P[V(\instance,\gamma(U), {\bf m}_{\gamma(U)}) \, \wedge \, S := \{ c \,|\, V(\instance, c, {\bf m}_{c}) \} \not\in \frak{S} ] \nonumber\\[0.5ex]
		&\leq P[\gamma(U) \in S := \{ c \,|\, V(\instance, c, {\bf m}_{c}) \} \not\in \frak{S} ] \nonumber\\[0.5ex]
		&\leq \max_{S  \not\in \frak{S}}P[\gamma(U) \in S ] \nonumber\\[-0.7ex]
		&\leq p_{triv}^{\mathfrak S} \, .\label{eq:Merkcase3}
		\end{align}
		By Theorem \ref{thm:simple}, we now get
		\begin{align}
		\QTC{ \SZ \backslash \SUC \backslash \CL}{\SUC}&\leq \max_{x,D}\sqrt{10 P\bigl[U \!\in\! \L^{x,D} \bigr]} \nonumber\\
		&\le\sqrt{10}\sqrt{ \max\left(\frac{2\ell-1}{|{\cal Y}|},\frac{q(2\ell - 1)}{|{\cal Y}|}, p_{triv}^{\mathfrak S}\right)}\nonumber\\
		&\le\sqrt{10}\sqrt{ \max\left(q \ell\cdot 2^{-n+1}, p_{triv}^{\mathfrak S}\right)},\nonumber
		\end{align}
		where we have used Equations \eqref{eq:Merkcase1}, \eqref{eq:Merkcase2} and \eqref{eq:Merkcase3} in the second inequality. Combining with Equations \eqref{eq:Merkcollision} and \eqref{eq:Merksplitthecapacity} yields the desired bound.
		\qed
	\end{proof}
	
	Similarly to Theorem~\ref{thm:extractor}, we now obtain the following. 
	
	\begin{theorem}\label{thm:extractor-M}
		Let $\Pi$ be an $\frak{S}$-sound$^*$ \MCnO protocol with challenge space ${\cal C}_\lambda$. Then $\sf FS[\Pi]$ is a \PoKOE in the QROM (as in Definition \ref{def:PoKOnline}), with $\varepsilon_\text{\rm sim}(\lambda,q,n) = 0$ and 
		\shortorfullversion{
		\begin{align*}
			\varepsilon_\text{\rm ex}(\lambda,q,n)  &\leq 2(\kappa\log \ell+1) \cdot 2^{-n} \!\!+ \Big(2eq^{3/2}2^{-n/2} \!\!+ q\sqrt{ 10 \max\left(q \ell\cdot 2^{-n+1}, p_{triv}^{\mathfrak S}\right)} \Big)^2 \\
			&\leq \left(22\ell\log \ell+60\right)q^32^{-n}+20 q^2 p_{triv}^{\mathfrak S}
		\end{align*}
	}{
	\begin{align*}
			\varepsilon_\text{\rm ex}(\lambda,q,n)  &\leq 2(\kappa\log \ell+1) \cdot 2^{-n} \!\!+ \bigg(2eq^{3/2}2^{-n/2} \!\!+ q\sqrt{ 10 \max\left(q \ell\cdot 2^{-n+1}, p_{triv}^{\mathfrak S}\right)} \bigg)^2 \\
			&\leq \left(22\ell\log \ell+60\right)q^32^{-n}+20 q^2 p_{triv}^{\mathfrak S}
		\end{align*}
	}%
	where $\kappa = \kappa(\lambda)  := \max_{c\in{\cal C}_\lambda }|c|$ and $\ell$ is the number of leaves of the Merkle-tree-based commitment.
	The running time of the extractor is dominated by running the compressed oracle, which has complexity $O(q^2) \cdot poly(n,B)$, and by computing $\MRoot_D^{-1}(y)$ and running~${\cal E}^*$.
	\end{theorem}
	Here again the proof follows exactly the outline of its counterpart from Section \ref{sec:OnExFS}, with some minor alterations to cope with the formalism of a Merkle-tree based C\&O \sigp. The difference in the bound is simply due to the difference between Lemmas \ref{lem:SUCCLbound} and \ref{lem:SUC-Oc-CLbound}. 
	\onlyshortversion{We refer to Appendix \ref{App:OmittedProofs} for the full proof. }

\begin{proof}
	We consider an arbitrary but fixed $\lambda \in \N$.
	Let $\cal P^*$ be a dishonest prover that, after making $q$ queries to a random oracle $H$, outputs and instance $\instance$ and a proof $\pi = (y,{\bf m}_\circ,O)$ plus some (possibly quantum) auxiliary output $Z$, where $O$ is an authentication octopus as defined in Section \ref{sec:Octopus}. For simplicity, we assume that $|c|$ is the same for all $c \in {\cal C}_{\lambda}$, and thus equal to $\kappa$. If it is not, we could always make the prover output a couple of dummy outputs $m_i$ to match the upper bound on $|c|$. In the experiment ${\cal V}^{\cal E} \circ {\cal P}^*{}^{\cal E}(\lambda)$, our extractor $\cal E$ works as follows while simulating all queries to $H$ (by $\cal P^*$ and $\cal V$) with the compressed oracle:
	\begin{enumerate}
		\item Run $\cal P^*(\lambda)$ to obtain $(\instance,\pi,Z)$ with $\pi=(y,{\bf m}_\circ, O)$. 
		\item Compute $v \leftarrow {\cal V}^H(\instance,\pi)$, given by the truth value of 
		$$
		\OctoVerify^H(c,y,{\bf m}_\circ,O) \quad \wedge \quad V(\instance,c,{\bf m}_\circ)  \quad \text{with} \quad c := \gamma({H}(\instance,y)) \, .
		$$
		\item Measure the internal state of the compressed oracle to obtain $D$.
		\item Run ${\cal E}^*$ on input $\MRoot_D^{-1}(y)$ to obtain $w$. 
	\end{enumerate}

	Note that in the views of both  $\cal P^*$ and $\cal V$, the interaction with $H$ and the interaction with $\cal E$ differ only in that their oracle queries are answered by a compressed oracle instead of a real RO in the latter case. This simulation is perfect and therefore $\varepsilon_\text{\rm sim}(\lambda,q,n) = 0$.

	Considering ${\cal P}^*$ as the algorithm $\cal A$ in Corollary~\ref{cor:linkgen}, the composition ${\cal V} \circ {\cal P}^*$ then matches up with the algorithm $\tilde{\cal A}$ for ${\cal F} = {\cal V}$. 
	Thus, noting that $\kappa(\log\ell + 1)$ is an upper bound on the amount of queries that $\OctoVerify$ makes, 
	\begin{align*}
	\Pr\bigl[ v \neq {\cal V}^D(\instance,\pi) \bigr] \leq 2(\kappa\log\ell+1)\cdot 2^{-n} \, .
	\end{align*}
	Therefore, we can bound bound the figure of merit $\varepsilon_\text{\rm ex}$ as
	\begin{align*}
	\varepsilon_\text{\rm ex}(\lambda,q,n) &= \Pr\bigl[v = 1 \,\wedge\, (\instance,w)\notin R\bigr] \\
	&\leq \Pr\bigl[{\cal V}^D(\instance,\pi) \,\wedge\, (\instance,w)\notin R\bigr] + 2(\kappa\log\ell+1)\cdot 2^{-n}  \\
	&\leq \Pr[{\cal V}^D(\instance,\pi) \,\wedge\, (\instance,w)\notin R \,|\, D \not\in \SUC \cup \CL] \\
	&\qquad+ \Pr[D \in \SUC \cup \CL] + 2(\kappa\log\ell+1)\cdot 2^{-n} \, .
	\end{align*}
	Using the definition of ${\cal V}^D(\instance,\pi)$, understanding that $c := \gamma \circ D(\instance, y)$, we can write the first term as 
	\begin{align*}
	\Pr\bigl[ & \OctoVerify^D(c,y,{\bf m}_\circ,O)
	\, \wedge \, V(\instance, c, {\bf m}_\circ) \,\wedge\, (\instance,w)\notin R \,|\, D \not\in \SUC \cup \CL \bigr]\\
	&\leq \Pr\bigl[ V(\instance, c, {\bf m}_c) \text{ for } {\bf m}:= \MRoot_D^{-1}(y) \,\wedge\, (\instance,w)\notin R \,|\, D \not\in \SUC \cup \CL \bigr]\\
	&\leq\Pr\bigl[D\in\SUC \,|\, D \not\in \SUC \cup \CL \bigr]\\
	&= 0 \, ,
	\end{align*}
	where the first equality exploits that $D(m) = h$ iff $m = D^{-1}(h)$ for $D \not\in \CL$. 
	
	We may thus conclude that
	\begin{align*} 
	\varepsilon_\text{\rm ex}(\lambda,q,n)
	&\leq 2(\kappa\log\ell+1)\cdot 2^{-n}\cdot 2^{-n} + \Pr\bigl[D\in \SUC\cup \CL \bigr] \\
	&\leq 2(\kappa\log\ell+1)\cdot 2^{-n} + \qQTC{\bot}{\SUC \cup \CL}^2  \, ,
	\end{align*}
	where the last inequality is by definition of $\qQTC{\bot}{\cdot\,}$. The claimed bound now follows from Lemma~\ref{lem:SUC-Oc-CLbound}.
	\qed
	\end{proof}

	\subsection{Discussion: Application to Picnic, and Limiting the Proof Size} \label{sec:picnic}
	\subsubsection{Application to Picnic.}
	
	A prominent use case of \CnO protocols is the construction of digital signature schemes via the \FST. An important example is Picnic \cite{Chase2017} currently under consideration as an alternate candidate in the NIST standardization process for post-quantum cryptographic schemes \cite{NIST}. On a high level, the design of Picnic can be described as follows. A \CnO \sigp is constructed using the MPC-in-the-head paradigm \cite{IKOS07}. Then, the \FST is applied in the usual way to obtain a digital signature scheme. There are three evolutions of Picnic: Picnic-FS, Picnic 2 and Picnic 3.\footnote{\shortorfullversion{There is also a version using the Unruh transformation.}{The original evolution also came with a variant using the Unruh \trafo, Picnic-Ur. We restrict our attention to the variants using the \FST.}} Picnic-FS 
	uses plain hash-based commitments, while Picnic 2 and Picnic 3 use a Merkle-tree-based collective commitment.
	
	All three evolutions enjoy provable post-quantum security when the hash function used for the \FST is modeled as a (quantum-accessible) RO. The best reduction applying to all of them proceeds as follows. First, Unruh's rewinding lemma \cite{Unruh2012} is used to construct a knowledge extractor for the underlying \sigp based on an appropriate $\mathfrak S$-soundness notion. Then, the \emph{generic} QROM reduction for the \FST from \cite{DFMS19} is used to construct a knowledge extractor for the signature scheme in the QROM from the extractor for the \sigp. Finally, the technique from \cite{GHHM20} is used for simulating the chosen-message oracle to reduce breaking NMA (no-message attack) security to breaking CMA (chosen-message attack) security. This final step connects to the previous one because for the signature scheme the witness extracted from an NMA attacker is the secret key.
	
	The first two steps\onlyfullversion{ in this chain of reductions}, i.e. Unruh's rewinding and \cite{DFMS19}, are\onlyfullversion{, however,} not tight: The former loses at least a fifth power in the Picnic case, and the latter a factor of $q^2$, where $q$ is the number of RO queries. This means that an NMA attacker with success probability $\epsilon$ can be used to break the underlying hard problem with probability $\Omega(\epsilon^5/q^{10})$ 
	(or worse, depending on the Picnic variant).
	
	For Picnic-FS (only), when in addition modeling the hash function used for the commitments as a RO, Unruh's rewinding can be replaced with \onlyfullversion{the} tight online extraction \onlyfullversion{technique} from \cite{DFMS21}. The remaining loss due to the \FS reduction is of order $\epsilon/q^2$, up to some additive terms accounting for search and collision finding in the RO, a sizable improvement over the above but still not tight.
	
	By analyzing the \FST of a \CnO protocol (with or without Merkle tree commitments) directly, our results provide a tight alternative to the above lossy reductions. 	Using Theorems \ref{thm:extractor} (for Picnic-FS) and \ref{thm:extractor-M} (for Picnic 2 and Picnic 3) we can avoid all multiplicative/power losses in the reduction for NMA security. An NMA attacker with success probability $\epsilon$ can\shortorfullversion{ thus}{, in other words,} be used to break the underlying hard problem with probability $\epsilon$, up to \onlyfullversion{some} unavoidable additive terms \shortorfullversion{due to}{accounting for} search and collision finding in the RO.

	\subsubsection{An observation about octopus opening sizes.}
	Depending on the parameters of the \CnO protocol, the octopus opening information, $\MOcto(c,\mathbf m)$ can be \shortorfullversion{much}{significantly} smaller than the concatenation of the individual authentication paths. On the other hand, it is also \emph{variable in size} (namely dependent on the choice of the challenge $c$), and the variance can be significant (see e.g. the computations for gravity SPHINCS in \cite{AE18}). In the context of a digital signature scheme constructed via the \FST of a \MCnO protocol, like, e.g., Picnic 2 and Picnic 3, this leads to the undesirable property of a variable signature size, where signatures can be quite a bit larger in the worst case than on average. This might, e.g., lead to problems when looking for a drop-in replacement for quantum-broken digital signature schemes for use in a larger protocol, where signatures need to be stored in a data field of fixed size.
	
	One option to mitigate this situation is to cut off the tail of the octopus size distribution, i.e. to restrict the challenge space of the \MCnO protocol to \onlyfullversion{the set of} challenges whose octopus is not larger than some bound. This can be done before applying the \FST, e.g. using rejection sampling. In that way, one obtains a digital signature scheme with significantly reduced worst case signature size, at the expense of a tiny security loss.

	\subsection{The Merkle-Tree-Based Unruh \Trafo}
	
	The Merkle tree based commitment mechanism can replace plain RO based commitments in {\em any} ordinary \CnO protocol, in particular in $\Pi := \mathsf{pU}[\Sigma]$ for any \sigp $\Sigma$. 
	The result is a \MCnO protocol and we obtain a corollary analogous to Corollary \ref{cor:Unruh}. 
	\begin{corollary}\label{cor:MUnruh}
		Let $\Sigma$ be an $\mathfrak S$-sound \sigp 
		with challenge space size $\ell_{0}$. Then $\mathsf{FS}[\mathsf{MPpU}_r[\Sigma]]$ is online-extractable with
		\begin{equation}
			\varepsilon_{\mathrm{ex}}\le  \left(22r\ell_0\log\left(r\ell_0\right)+60\right)q^32^{-n}+20 q^2 \left(p_{triv}^{\mathfrak S}\right)^r
		\end{equation}
		where $\mathsf{MPpU}_r[\Sigma]$ is the \textbf{M}erkle-tree-based, \textbf{P}arallel-repeated, \textbf{p}re-\textbf{U}nruh \trafo of $\Sigma$, i.e., the \MCnO protocol obtained by replacing the commitments of $\mathsf{pU}[\Pi]^r$ with a Merkle-tree-based collective commitment.
	\end{corollary}
	

	\section{Acknowledgement}
	JD was funded by ERC-ADG project 740972 (ALGSTRONGCRYPTO).
	CM was funded by a NWO VENI grant (Project No. VI.Veni.192.159). 
	CS was supported by a NWO VIDI grant (Project No. 639.022.519).

	\bibliographystyle{alpha}
	\bibliography{QROM}

	\appendix
	
	\section*{APPENDIX}

	\section{$\frak{S}$-Soundness for Plain $\Sigma$-Protocols}\label{sec:S-soundOrd}
	
	Similar to Definition~\ref{def:S-sound-star}, for plain \sigps (i.e., \sigps in the standard model) the standard notion of special-soundness and $k$-soundness generalize as follows. Also here, $\frak{S}$ is a non-empty, monotone increasing set of subsets $S \subseteq \cal C$, and  $\frak{S}_{\min} := \{S \in \frak{S} \,|\, S_\circ \subsetneq S \Rightarrow S_\circ \not\in \frak{S} \}$, but now, a challenge $c \in \cal C$ is not (necessarily) a subset of $[\ell]$ anymore. 
	
	\begin{definition}\label{def:S-sound-Ord}
		A \sigp $\Pi$ is called $\frak{S}$\emph{-sound} if there exists an efficient deterministic algorithm $\mathcal E_{\frak{S}}(\instance,a,S,\{z_c\}_{c \in S})$  that takes as input an instance~$\instance$, a first message $a$, a subset $S \subseteq \cal C$ of challenges, and responses $z_c$ for $c \in S$, and it outputs a witness for $\instance$ if $S \in \frak{S}_{\min}$ and ${\cal V}(\instance, a, c,z_c)$ for all $c \in S$. 
		\footnote{We note the clash in terminology with Definition \ref{def:S-sound}. However, Definition \ref{def:S-sound} applies exclusively to \CnO \sigps in the (Q)ROM, whereas the definition here applies exclusively to \sigp in the plain model; so there should be no confusion. The two definitions are of course related: a $\frak{S}$-sound \CnO \sigp becomes a $\frak{S}$-sound plain \sigp when the commitments are instantiated with a perfectly binding commitment scheme (rather than with a hash function). }
	\end{definition}
	
	
	The common notion of a {\em special-sound} \sigp is then the special case of a $\frak{S}$-sound \sigp with $\frak{S} := \{S \subseteq {\cal C}  \,|\, |S| \geq 2 \}$, and similarly a {\em $k$-sound} \sigp is a $\frak{S}$-sound \sigp with  $\frak{S} := \{S \subseteq {\cal C}  \,|\, |S| \geq k \}$. 
	Also here, using syntactically the same definition as in Equation~(\ref{eq:ptriv}),
	\begin{align*} 
		p^{\mathfrak S}_{triv} := \frac{1}{|\mathcal{C}|} \max_{\hat S \not\in  \frak{S}} |\hat S| 
	\end{align*}
	then captures the ``trivial'' attack that may potentially (and typically does) apply to a $\frak{S}$-sound \sigp, where the dishonest prover prepares a first message $a$ so that he has valid responses $z$ ready for all the challenges $c$ in some arbitrarily chosen set $\hat S \not\in  \frak{S}$.

\end{document}